\documentclass[conference]{IEEEtran} 

\IEEEoverridecommandlockouts 

\usepackage{ifthen}
\usepackage{tikz,filecontents}
\usetikzlibrary{shapes,arrows,shadings,patterns}
\newlength\figureheight
\newlength\figurewidth
\usepackage{algorithmic}
\usepackage[noend,lined,linesnumbered]{algorithm2e} 
\makeatletter
\renewcommand{\@algocf@capt@plain}{above}
\makeatother
\usepackage{amsthm}
\usepackage{amsmath}

\newtheorem{proposition}{Proposition}
\newtheorem{fact}{Fact}

\theoremstyle{definition}
\newtheorem{definition}{Definition}
\theoremstyle{definition}

\usepackage[caption=false]{subfig}

\usepackage{amsfonts}
\usepackage{multirow}

\usepackage{tikz}
\usetikzlibrary{arrows, calc, patterns}

\tikzset{
  treenode/.style = {align=center, inner sep=0pt, text centered,
    font=\sffamily},
  arn_n/.style = {treenode, circle, white, font=\sffamily\bfseries, draw=black,
    fill=black, text width=1.5em},
  arn_r/.style = {treenode, circle, red, draw=red, 
    text width=1.5em, very thick},
  arn_x/.style = {treenode, rectangle, draw=black,
    minimum width=0.5em, minimum height=0.5em}
}

\usepackage{xcolor}
\definecolor{comgray}{HTML}{BEBEBE}
\definecolor{citegreen}{HTML}{458B00}
\usepackage{hyperref}
\hypersetup{
   colorlinks=true,
   citecolor=citegreen
}

\usepackage{url}

\usepackage{booktabs} 

\begin{document}

\title{Private Continual Release of Real-Valued Data Streams$^*$\thanks{*This is a preliminary version of the paper with the same title which to appear in the proceedings of the Network and Distributed System Security Symposium, NDSS 2019.}}

%
%

\author{
  \IEEEauthorblockN{Victor Perrier}
  \IEEEauthorblockA{ISAE-SUPAERO\\
  \& Data61, CSIRO\\
  v.perrier@gmail.com}
  \and
  \IEEEauthorblockN{Hassan Jameel Asghar}
  \IEEEauthorblockA{Macquarie University\\ \& Data61, CSIRO\\
  hassan.asghar@mq.edu.au}
  \and
  \IEEEauthorblockN{Dali Kaafar}
  \IEEEauthorblockA{Macquarie University\\ \& Data61, CSIRO\\
  dali.kaafar@mq.edu.au}
  }

 \makeatletter\def\@IEEEpubidpullup{7\baselineskip}\makeatother
 \IEEEpubid{\parbox{\columnwidth}{Permission to freely reproduce all or part of this paper for noncommercial purposes is granted provided that copies bear this notice and the full citation on the first page. Reproduction for commercial purposes is strictly prohibited without the prior written consent of the Internet Society, the first-named author (for reproduction of an entire paper only), and the author's employer if the paper was prepared within the scope of employment.  \\
 NDSS '19, 24-27 February 2019, San Diego, CA, USA\\
 Copyright 2019 Internet Society
 }
 \hspace{\columnsep}\makebox[\columnwidth]{}}
\maketitle

\pagestyle{plain}

\begin{abstract}
We present a differentially private mechanism to display statistics (e.g., the moving average) of a stream of real valued observations where the bound on each observation is either too conservative or unknown in advance. This is particularly relevant to scenarios of real-time data monitoring and reporting, e.g., energy data through smart meters. Our focus is on real-world data streams whose distribution is \emph{light-tailed}, meaning that the tail approaches zero at least as fast as the exponential distribution. For such data streams, individual observations are expected to be concentrated below an unknown threshold. Estimating this threshold from the data can potentially violate privacy as it would reveal particular events tied to individuals~\cite{meter-memoir}. On the other hand an overly conservative threshold may impact accuracy by adding more noise than necessary. We construct a utility optimizing differentially private mechanism to release this threshold based on the input stream. Our main advantage over the state-of-the-art algorithms is that the resulting noise added to each observation of the stream is scaled to the threshold instead of a possibly much larger bound; resulting in considerable gain in utility when the difference is significant. Using two real-world datasets, we demonstrate that our mechanism, on average, improves the utility by a factor of 3.5 on the first dataset, and 9 on the other. While our main focus is on continual release of statistics, our mechanism for releasing the threshold can be used in various other applications where a (privacy-preserving) measure of the scale of the input distribution is required.
\end{abstract}

\maketitle

\section{Introduction}
\label{sec:intro}

Many services can benefit from real-time monitoring of statistics from customer data. Examples include electricity usage in a neighbourhood collected through smart meters, customers' expenditure in a supermarket on a given day, and commute time of residents of a city during peak hours. Statistics for these applications can be obtained from real-time data collected through a variety of sensors and refreshed as new data arrives. These statistics can then be displayed to analysts and planners who could use them to optimize services. Privacy concerns, however, preclude release of raw statistics. For instance, a customer at a pharmacy would not be willing to disclose the purchase of medicines linked to a peculiar health condition. Likewise, analysis of smart meter data can likely reveal the activities of a particular household or even whether anyone is at home or not. Such privacy violations have been demonstrated for the case of smart-meter data where patterns such as the number of people in the household as well as sleeping and eating routines were revealed even without any prior training~\cite{meter-memoir}. The goal therefore is to enable monitoring of statistics without compromising individual privacy. 

A natural candidate for privacy protection is the rigorous framework of differential privacy~\cite{calib-noise, dp-book}. Informally, any algorithm satisfying the definition of differential privacy has the property that its output 
distribution (based on the coin tosses of the algorithm) on a given database is close in probability to the output 
distribution if any single row in the dataset is replaced. The closeness is parameterized by the privacy budget 
$\epsilon$. Most of the work on differential privacy has focused on static (input) datasets, and there has been 
very little focus on datasets that are continuously being updated as in our setting~\cite{cont-observe, cont-release}. 
Despite this, there is a growing need to shift focus to provide privacy in the dynamic setting which is likely 
to be more pervasive in the near future~\cite{information-rich}. 

More precisely, our scenario is concerned with releasing statistics from a sequence of observations arriving in a streaming fashion
each within some public upper bound $B$. Our statistic of interest is the continually changing average  
as new observations arrive. This can be readily obtained by summing all the 
observations seen thus far (since the number of observations is assumed public). We remark that our focus is on 
approaches that provide \emph{event level} privacy~\cite{cont-observe} only, which means that individuals are 
guaranteed that their peculiar events remain private but not necessarily the general trend.\footnote{The latter is guaranteed through \emph{user-level} privacy, i.e., privacy for all events from a user. See \cite[\S 12]{dp-book} and~\cite{cont-observe} for a further discussion on the merits of event versus user-level privacy.} For many use cases this 
is a suitable guarantee of privacy, e.g., individuals might be happy to disclose their routine trip to work while 
unwilling to share the occasional detour. One way to release the sum via differential privacy is to add independent noise generated through the Laplace 
distribution scaled to ${B}$~\cite{calib-noise}. However, this results in cumulative error (absolute difference from the true sum) of $O(B\sqrt{n})$ after $n$ observations. 
Two aforementioned works on continual release of datasets, i.e.,~\cite{cont-observe} and \cite{cont-release}, focus 
on binary streams, where each observation is either 0 or 1. We can generalize their algorithm to observations within the bound $B$ which results in a considerably reduced error of $O(B (\log_2 n )^{1.5})$. 


While this significantly reduces the error over the basic approach, the error is still proportional to $B$. In many real world situations, the bound $B$ might not be known in advance, or known only as the worse case bound resulting in an overly conservative estimate of the true bound. Likewise, perhaps most observations are tightly concentrated below an unknown threshold $\tau$ well below $B$. For instance, returning to our commute time use case, it is highly unlikely that anyone would be commuting for the full 24 hours on a given day. We are interested in a mechanism that allows us to determine a threshold $\tau$ below which majority of the observations are concentrated. This in turn allows to release statistics with noise scaled to $\tau$ rather than $B$ resulting in error $O(\tau (\log_2 n )^{1.5})$, which is a significant improvement depending on $B$, $\tau$ and $n$.\footnote{For instance, assume $n = 1,000,000$ and the known bound is $B = 10,000$, and we are interested in the average. Assume further that almost all observations are within $\tau = 100$ with an average of 30. Then, through the original mechanism we get the (noisy) average as $30 \pm 1$. Through the mechanism that scales noise according to $\tau$, we get the noisy average as $30 \pm 0.01$, an improvement by a factor of $B/\tau = 100$. This can be significant if the average is required with high precision.}

However, estimating $\tau$ is not straightforward due to a number of reasons. First, estimating $\tau$ beforehand would result in high or even unbounded cumulative error due to outliers. Thus, any algorithm needs to observe at least a small subset of initial observations before determining $\tau$. This \emph{time lag} needs to be optimised for accuracy: estimating $\tau$ too early will result in high accumulated error, and too late will only show marginal improvement over the default case (i.e., when using $B$ as the estimate). Likewise, again for reasons of accuracy, we need to ensure that readings outside the threshold are sporadic. Finally and most importantly, naively estimating $\tau$ can result in privacy violation by leaking information specific to an individual, e.g., if we take the maximum of the observations seen so far as $\tau$, we display the exact value corresponding to a particular event from an individual.  

In this paper, we propose a mechanism that allows us to estimate the threshold $\tau$ using a subset of observations from an incoming stream via differential privacy, simultaneously optimizing utility for releasing the moving average. Although we optimize utility for the case of moving averages, our mechanism for releasing the threshold is generic enough to be used for other statistics and applications. These include displaying the average with a sliding window~\cite{decay-ave} or releasing histogram of the streaming data~\cite{xu-hist} where in all cases the noise will be scaled to the most concentrated part of the distribution of the stream. 

In addition to theoretical accuracy guarantees, we provide empirical evidence of the utility gain of our scheme using two real world datasets: the first dataset contain about 50 million individual trip times on public trains in a major metropolitan city over a period of two weeks, and the second dataset is composed of individual amount spent over 140,000 transaction by about 1,000 customers in a major supermarket. Using the two datasets we first verify that real world data has the property that most readings are concentrated tightly well below a conceivable conservative bound $B$. Using the same datasets we then show that our improved algorithm displays the average statistic (commute time or amount spent) with a utility many orders of magnitude ($\approx 3.5$ and 9 resp., on the two datasets) better than applying (generalized versions of) the state of the art algorithms~\cite{cont-observe, cont-release}. Our utility gain is for data streams that obey a \emph{light-tailed distribution}, namely a distribution whose tail lies below the exponential distribution (beyond the above mentioned threshold; see Section~\ref{sub:stat-def} for a precise definition). We argue and show that real-world datasets are expected to satisfy this property.

\section{Background}
\label{sec:background}
In this section we formally describe our problem, associated definitions and overview of the algorithm from~\cite{cont-observe} and \cite{cont-release} referred to as the binary tree (BT) algorithm which will serve both as a benchmark and a sub-module of our technique. 

\subsection{Problem Statement} 

Let $B$ be a positive real number. We model input streams (or strings), denoted $\sigma$, as the set of finite strings $\Sigma = [0, B]^\mathbb{N}$ of length at most $n$. 
The $i$th element of $\sigma$ shall be denoted by $\sigma(i)$, and shall be called the $i$th observation or reading. A generic element or observation from $\sigma$ shall be denoted by $x$. For $j \ge i$, $\sigma(i {:} j)$ represents the substring (or sub-stream) $\sigma(i) || \cdots || \sigma(j)$, where $||$ is the concatenation operator. We are interested in finding the average of the elements of the stream $\sigma$ at each time step $i \in \mathbb{N}$. This reduces to finding $\sum_{j = 1}^i \sigma(j)$ at each step $i \in [n]$, since we assume the \emph{observation counter} to be public. Our goal is to release a privacy-preserving version of this sum. 



\subsection{Privacy Definitions}

\begin{definition}[Sum Query]

We call the function $c: \Sigma \times \mathbb{N} \rightarrow \mathbb{R}$ defined for $\sigma \in \Sigma$ and $i \in [n]$ as
$
c(\sigma, i) = \sum_{j=1}^i \sigma(i)
$
as the sum query. 

\end{definition}

%
%

\begin{definition}[Adjacent Streams]

Let $\sigma,\sigma' \in \Sigma$, The Hamming distance $d(\sigma,\sigma')$ is the number of elements different in the corresponding positions of the two strings, i.e.,
$
d(\sigma,\sigma') = |\{i : \sigma(i) \neq \sigma'(i),\forall i \in \mathbb{N} \} |.
$
The two streams $\sigma$ and  $\sigma'$ are adjacent if and only if $d(\sigma,\sigma') = 1$. 

\end{definition}

\begin{definition}[$(\epsilon,\delta)$-Differential Privacy]

A summation mechanism $M$ is $(\epsilon,\delta)$-differentially private if and only if for any two adjacent streams $\sigma, \sigma'$ we have $\forall n \in \mathbb{N}$ and $\forall S \subset \mathbb{R}$,
\[
\Pr \left[ M(c, \sigma, n) \in S \right] \leq \Pr \left[ M(c, \sigma', n) \in S \right] \times e^{\epsilon} + \delta,	
\]
where $\epsilon$ is a small constant and $\delta$ is a negligible function in $n$. We shall use $\hat{c}$ to denote the output of $M$ in the following.
\end{definition}

The privacy definition does not assume the stream $\sigma$ to have any specific distribution, barring the fact that each of its element is within $[0, B]$. For utility however we shall assume that the streams are sampled with some underlying probability distribution with support over the set $[0, B]$. 

\begin{definition}[Probability Distribution of Streams]
\label{def:dist-stream}
Let $B \in \mathbb{R}^+$. Denote by $\mathcal{F}_B$ the probability distribution which satisfies $\Pr \left[ X \in [0, B] \right] = 1$, for any random variable $X$ distributed as $\mathcal{F}_B$. A string $\sigma$ is said to have distribution $\mathcal{F}_B$, if for all $i \in \mathbb{N}$, $X_i = \sigma(i)$ is sampled from $\mathcal{F}_B$. We denote this by $\sigma \leftarrow_{\mathcal{F}_B} \Sigma$.
\end{definition}

\begin{definition}[$(\alpha, \beta)$ Utility]
\label{def:utility}
The mechanism $\hat{c}$ is said to be $(\alpha, \beta)$-useful if for all $n \in \mathbb{N}$ and $\sigma \leftarrow_{\mathcal{F}_B} \Sigma$,
\[
\Pr \left[ | \hat{c} (\sigma, n) - c( \sigma, n) | \le \alpha \right] \ge 1- \beta, 
\]
where the probability is over the coin tosses of $\hat{c}$ and the distribution $\mathcal{F}_B$. 
\end{definition}
Note that the above is different from the utility definition in~\cite{cont-release}, where the probability is over the coin tosses of $\hat{c}$ only, and hence the inequality is satisfied for all strings $\sigma$. In our case, we shall be utilizing the probability that certain strings are more likely realized in practice; hence the use of the distribution $\mathcal{F}_B$. We stress again that the privacy definition does not rely on $\mathcal{F}_B$. 

%
%


Consider an arbitrary function $c : \Sigma \rightarrow \mathbb{R}$. The sum query falls under this definition with an auxiliary parameter $n \in \mathbb{N}$. We first define the global sensitivity of $c$.
\begin{definition}[Global Sensitivity]
The global sensitivity of a function $c : \Sigma \rightarrow \mathbb{R}$, denoted $\mathsf{GS}$, is defined as
\[
\mathsf{GS}(c) =  \max_{\sigma, \sigma' \in \Sigma \text{ : } d(\sigma,\sigma') \le  1} |c(\sigma) - c(\sigma')|
\]
\end{definition}

\begin{definition}[Laplace Mechanism]
Let $\text{Lap}(b)$ denote the probability density function of the Laplace distribution with mean $0$ and scale $b$ given as
$
\text{Lap}(b) = \frac{1}{2b} \exp \left( \frac{|x| }{b} \right).
$
Then the mechanism 
$
\hat{c}(\sigma) = c(\sigma) + Y,
$
where $Y$ is drawn from $\text{Lap}(\frac{\mathsf{GS}}{\epsilon})$ is $(\epsilon, 0)$-differentially private~\cite{calib-noise}. 
\end{definition}

%
%
%

A definition of sensitivity that is defined for a particular input string $\sigma$ is called local sensitivity.
\begin{definition}[Local Sensitivity]
The local sensitivity of a function $c : \Sigma \rightarrow \mathbb{R}$ at $\sigma \in \Sigma$, denoted $\mathsf{LS}_\sigma$,  is defined as
\[
\mathsf{LS}_{\sigma} (c) =  \max_{\sigma' \in \Sigma \text{ : } d(\sigma,\sigma') =  1} |c(\sigma) - c(\sigma')|
\]
\end{definition}
The advantage of using local sensitivity is that we only need to consider neighboring strings of $\sigma$ which could result in lower sensitivity of the function $c$, and consequently lower noise added to the true answer $c$. Unfortunately, replacing the global sensitivity with local sensitivity naively in the Laplace mechanism (for instance) may not result in differential privacy~\cite{salil-tut}. This drawback can be removed by using smooth sensitivity~\cite{smooth} instead. 

\begin{definition}[Smooth Upper Bound]
For $b > 0$, an $b$-smooth upper bound on $\mathsf{LS}_{\sigma}$, denoted $\mathsf{SS}^*_{\sigma}$ satisfies:
\begin{align*}
\mathsf{SS}^*_{\sigma} (c) &\geq \mathsf{LS}_{\sigma}(c), \; \forall \sigma \in \Sigma,\\
\mathsf{SS}^*_{\sigma} (c) &\leq e^b \mathsf{SS}^*_{\sigma'}(c), \; \forall \sigma, \sigma' \in \Sigma : d(\sigma, \sigma') = 1
\end{align*}
\end{definition}

\begin{definition}[Smooth sensitivity]
\label{def:smooth}
For $b > 0$, the $b$-smooth sensitivity of $c$, denoted $\mathsf{SS}_{\sigma, b}(c)$, at $\sigma \in \Sigma$ is
\[
\mathsf{SS}_{\sigma, b}(c) = \max_{\sigma' \in \Sigma} \left\{ \mathsf{LS}_{\sigma'}(c) \cdot e^{-b d(\sigma, \sigma')} \right\}
\]
\end{definition}
Note that smooth sensitivity is the smallest function to satisfy the definition of a smooth upper bound~\cite{smooth}. Smooth sensitivity allows us to add noise proportional to $\frac{\mathsf{SS}_{\sigma, b}}{a}$ to the output of the function $c$ to obtain $(\epsilon, \delta)$-differential privacy. The choice of $a$ and $b$ depends on the privacy parameters and the distribution used to generate noise.  

\subsection{Statistical Definitions}
\label{sub:stat-def}
\begin{definition}[$p$-Quantile]
\label{def:p-quant}
Let $F$ be a cumulative distribution function (CDF) of some continuous random variable $X$. The $p$-quantile of $F$, denoted $x_p$, is defined as
\[
x_p = \inf \{ x \in \mathbb{R} : F(x_p) = \Pr( X \le x_p ) \ge 1 - p \}.
\]
\end{definition}

\begin{fact}
\label{fact:quant-func}
Let $X$ be an exponentially distributed random variable. Then its CDF is given by 
\[
H(x; \gamma) = \begin{cases}
			1 - e^{-\gamma x},  & x \ge 0,\\
			0, & x < 0.
\end{cases}
\]
Let $0 \le p < 1$. The quantile function of $H$ is given as
\begin{equation}
\label{eq:quant-func}
H^{-1}(p; \gamma) = -\frac{\ln p}{\gamma}
\end{equation}
\end{fact}

\begin{definition}[Light-tailed distribution]
\label{def:light-tail}
Let $X$ be a random variable with CDF $F$ and let $Y$ be an exponentially distributed random variable with CDF $H(\cdot; \gamma)$. Let $x_p$ be the $p$-quantile of $F$. Let $\gamma = -\frac{\ln p}{x_p}$, so that the $p$-quantile of $H$, i.e., $y_p$, is equal to $x_p$. We say that $X$ has a \emph{light-tailed distribution beyond $x_p$}, or equivalently $F$ is light-tailed beyond $x_p$, if $\forall x \ge x_p$,  $F(x) \ge H(x; \gamma)$.
\end{definition}
The choice $\gamma = -\frac{\ln p}{x_p}$ is immediate from Eq.~\ref{eq:quant-func}. 
\begin{proposition}
\label{prop:pr}
Let $X$ be exponentially distributed with CDF $H(\cdot; \gamma)$. Let $r \ge 1$. Let $x_p$ be the $p$-quantile of $H$ and let $x_{p^r}$ be the $p^r$-quantile of $H$. Then
\begin{equation}
\label{eq:pr}
x_p \cdot r \ge x_{p^r}
\end{equation}
\end{proposition}
\begin{proof}
When $r = 1$, we trivially have $x_{p^r} = x_p = x_p \cdot 1$. So, consider $r > 1$ and assume to the contrary that  $x_p \cdot r < x_{p^r}$. From Eq.~\ref{eq:quant-func}, this implies that
\begin{align*}
	-\frac{\ln p^r}{\gamma} & < -\frac{\ln (pr)}{\gamma} \\
	\Rightarrow -\ln p^r &< -\ln (pr) \\
	\Rightarrow	-r \ln p &< -\ln (pr) < -\ln p,
\end{align*}
which implies $r < 1$, a contradiction.
\end{proof}

\begin{proposition}
\label{prop:p-theta-m}
Let $X$ be a random variable with CDF $F$. Let $x_p$ be the $p$-quantile of $F$. The expected number of samples required to observe at least a constant number of samples $x \in X$ such that $x \ge x_p$ is $\Omega(\frac{1}{p})$. 
\end{proposition}
\begin{proof}
This follows from the properties of the binomial distribution. The probability that a sample $x \in X$ satisfies $x \ge x_p$ is given by $p$.  In $m$ samples, we expect $mp$ successes. Setting $mp \ge c$ for some constant $c$ gives us $m = \Omega(\frac{1}{p})$.
\end{proof}

%
%
%
%
%
%
%
%

\subsection{The Binary Tree (BT) Algorithm}
To release the private version of the sum $c(\sigma, i)$ at step $i \in [n]$, we shall use the binary tree algorithm from~\cite{cont-observe, cont-release} as a building block. We call this the BT algorithm. We briefly outline the algorithm, and discuss what we would like to improve. Given a string of length $n$, the BT algorithm first constructs a complete binary tree: the leaves are labelled by the intervals $[1, 1], [2, 2], \ldots, [n, n]$ and each parent node is the union of intervals of its two child nodes. To output the noisy count $\hat{c}(\sigma, i)$, the algorithm finds at most $\log_2 n$ nodes in the binary tree, whose union equal $[1, i]$. Thus, instead of adding noise of scale $\frac{B n}{\epsilon}$ (for $\epsilon$-differential privacy through a simple application of the Laplace mechanism), the noise added to each node is only scaled to $\frac{B\log_2 n}{\epsilon}$, resulting in an $\epsilon$-differentially private algorithm. For more concrete details, see~\cite{cont-observe, cont-release}.



\begin{figure}[ht]

\centering
\includegraphics[scale=0.35]{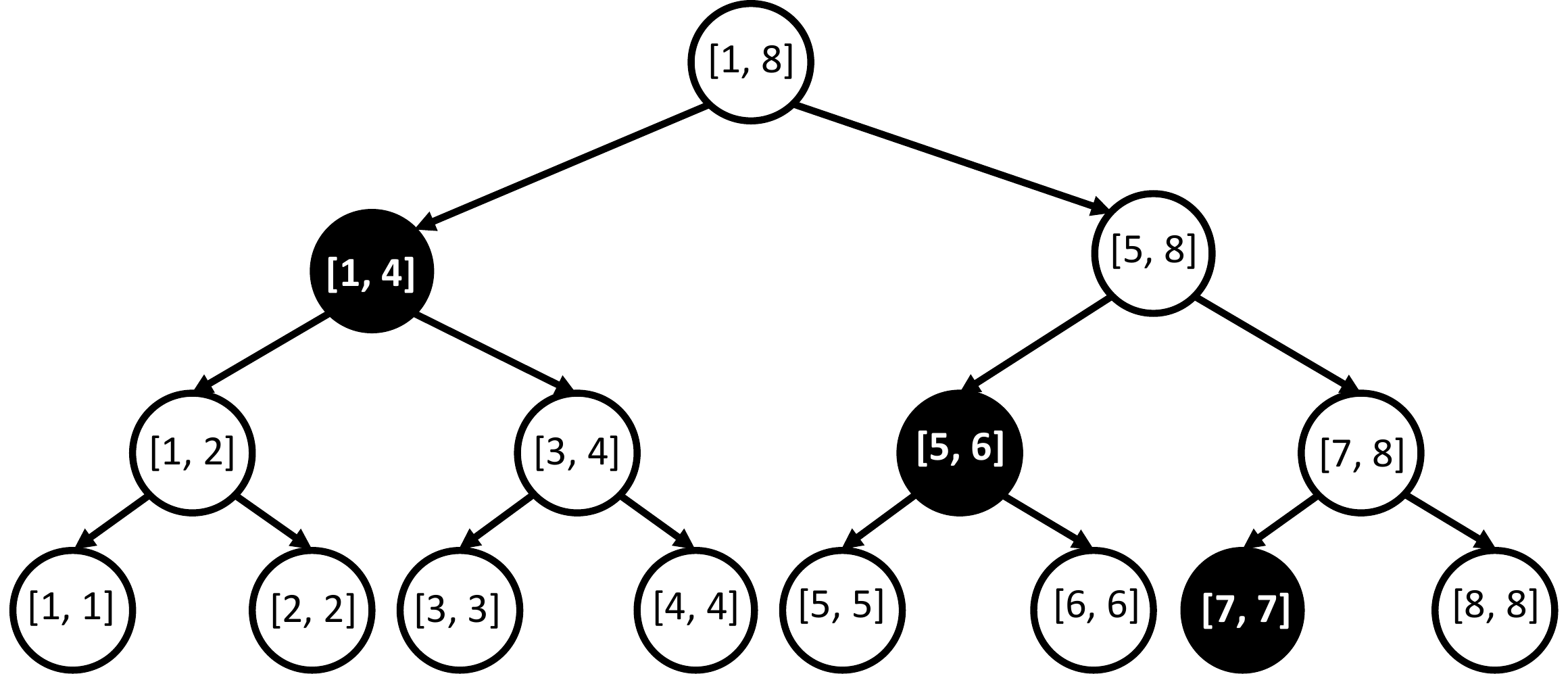}
\caption{Example of the binary tree algorithm~\cite{cont-observe, cont-release}. To compute the sum of the first 7 observations, noise is added to only 3 nodes whose union equals the interval $[1, 7]$.}
\label{fig:binary-tree}
\end{figure}

Figure~\ref{fig:binary-tree} illustrates the algorithm with an example. We have $n = 8$, and we wish to find $\hat{c}(\sigma, 7)$. This can be done by adding only three noisy sums corresponding to the nodes $[1, 4]$, $[5, 6]$ and $[7, 7]$ (shaded in the figure). That is
\begin{align*}
& \hat{c}(\sigma, 7) = \sum_{i=1}^4 \sigma(i) + \text{Lap}\left(\frac{B\log_2 8}{\epsilon}\right) + \sum_{i=5}^6 \sigma(i) \\
&+ \text{Lap}\left(\frac{B\log_2 8}{\epsilon}\right) + \sigma(7) + \text{Lap}\left(\frac{B\log_2 8}{\epsilon}\right)
\end{align*}



\subsection{Goal}
If the stream $\sigma$ has distribution $\mathcal{F}_B$ (see Definition~\ref{def:dist-stream}), then the global sensitivity of the function $c$ is $\mathsf{GS} = B$. 
For strings of length $n$, the BT algorithm has error~\cite{cont-release}
\begin{equation}
\label{eq:bt-alpha}
\alpha = O \left( \mathsf{GS} \cdot \frac{1}{\epsilon} \sqrt{8 \ln \frac{1}{\beta} } (\log_2 n)^{1.5} \right),
\end{equation}
with probability at most $\beta$. Since $\mathsf{GS} = B$, we get a linear term in $B$. We aim to improve the dependence on $B$. 
Our gain is on input streams $\sigma$ whose distribution is light-tailed beyond a threshold $\tau \ll B$ (see Definition~\ref{def:light-tail}). In other words, input streams whose distribution is concentrated below $\tau$. Then, instead of using global sensitivity, we will use smooth sensitivity tailored to the threshold $\tau$. This means that the noise added will be proportional to $\tau$ rather than $B$ in the BT algorithm, resulting in improved utility. In the next section, we give several examples of real-world datasets which are light-tailed. Thus, our improved approach has practical utility gains.



\section{Motivation and Overview of the Proposed Mechanism}
\label{sec:algo}

\subsection{Motivation}
There are many scenarios where an upper bound $B$ on a generic element of the stream is overly conservative:
\begin{itemize}
	\item There might not be a natural bound $B$ known in advance, e.g., a bound on the expenditure during a trip to the supermarket. Any guess on the bound $B$ would be taking into account instances of unusually high spendings. This will result in a very conservative upper bound.
	\item In some cases, a natural bound $B$ may exist. For instance, public transport commute time per day has a natural bound of $B = 24 \text{ hours}$. However, most commute times will be tightly concentrated well below this $B$. Once again, we would have an overly conservative estimate. 
\end{itemize}
Thus, our aim is to obtain a more realistic threshold $\tau \ll B$ tailored to the input stream.

%

\subsection{Empirical Validation of Concentration of Data}
\label{sub:datasets}
Here we validate our assumption that real data is concentrated well below a conceivable bound $B$. 
We consider two real world datasets: 

\begin{itemize}
	\item \emph{Train trips dataset:} This consists of commute times of trips made by passengers through public trains in the greater region surrounding a major metropolitan city.\footnote{We obtained this dataset from the relevant transport authority. The commute time per trip also includes any waiting time after and before a passenger has tapped on and off through the smart ticketing system.}
The aim here is to display the average commute time in real-time (using the cumulative sum). Informally, the privacy issue here is to hide the exact travel time of any single trip of an individual, as it may lead to inferring the individual's exact location at a given time. The total number of train trips in the dataset is about 50 million (spanning over 4 weeks).
%
	\item \emph{Supermarket expenditure dataset:} A dataset showing the amount of money spent by customers to a major supermarket retailer. The goal is to show real-time average expenditure. Informally, privacy property here is to hide the exact transaction amount of a customer on a trip to the supermarket, which may disclose the type of products bought by the customer. 
	This dataset is much smaller and contains about 140,000 transactions (a transaction contains multiple purchased items) by approximately 1,000 customers over a period of one year. 
\end{itemize}
The distribution of both datasets, when viewed as input streams, satisfies our definition of a light-tailed distribution (cf. Definition~\ref{def:light-tail}). 
This is depicted in Figure~\ref{fig:cdf-train}. The left graph shows the (smoothed) empirical cumulative distribution function (ECDF) of the time taken by a trip by a passenger in the train trips dataset.
We can see that the peak is around 20 to 30 mins, and very few customers take more than 150 minutes on a given day. Notice that this is significantly less than the maximum possible time in a 24 hour period (i.e., 1,440 minutes).
Likewise, we see a similar trend in the total expenditure during a trip to the supermarket on a given day, shown in Figure~\ref{fig:cdf-train} (right). 
Many other similar datasets are expected to have a light tailed distribution, e.g., phone-call durations or smart electricity meter readings.

\begin{figure}[ht]
\centering

\includegraphics[scale=0.4, trim={0 0 0 1mm},clip]{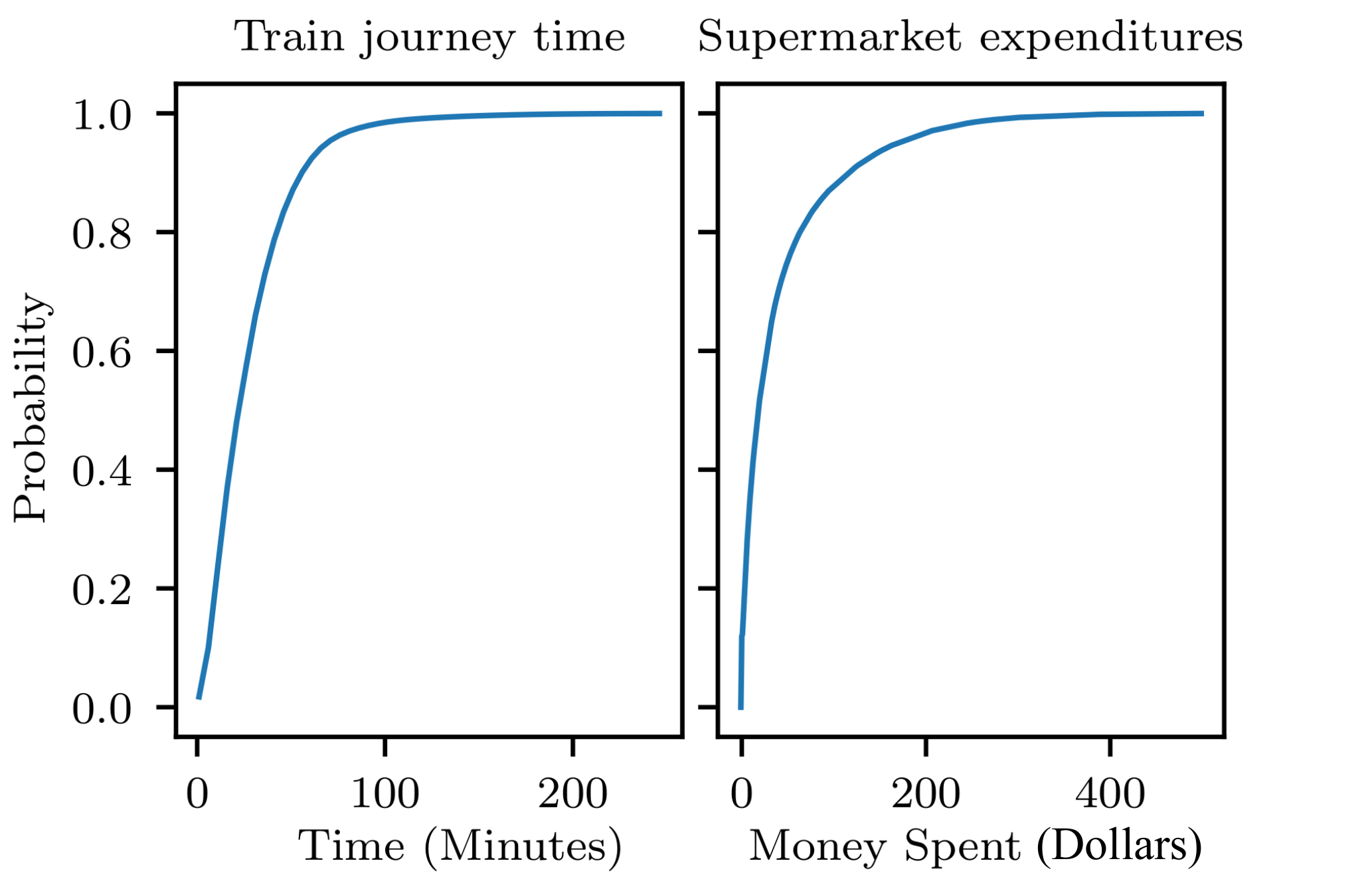}
\caption{ECDF of journey times (in minutes) from the train trips dataset (left) and expenditures (in dollars) from the supermarket dataset (right). Both are concentrated well below any conservative bound.}
\label{fig:cdf-train}
\end{figure}

\subsection{Overview of the Mechanism}
\label{sub:overview}
Figure~\ref{fig:overall-pic} shows the pictorial overview of our approach. We will first (privately) estimate a threshold $\tau$ using the first $m < n$ observations. We then release the sum of the first $m$ observations, i.e., $c(\sigma, m)$, at once using an application of the simple Laplace mechanism with noise scaled to $\approx \tau$. For each step after $m$, we continually release the sum $c(\sigma, i)$ for $m < i \le n$ using the BT algorithm with Laplace noise once again scaled to $\approx \tau$. As a consequence, we withhold releasing the sum before we have enough data points, quantified by $m$, from the stream to estimate $\tau$. We shall call $m$ the \emph{time lag}. Globally, there are two sources of error that we seek to minimise. First is the \emph{outlier error} (denoted $\alpha_{\text{out}})$: any readings above $\tau$ will be stripped to $\tau$ (before adding noise). This error can occur with probability $\beta_{\text{out}}$ indicated in the figure. The second is the accumulated error at the last step, i.e., $n$, due to Laplace noise whose probability is denoted $\beta_{\text{Lap}}$ in the figure. 

%


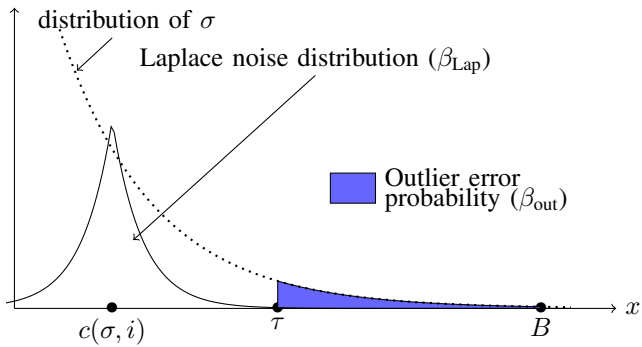
\begin{figure}[h]
\centering

\begin{tikzpicture}
\clip (-0.1,-1) rectangle (8.8,4.1);

\draw[->] (0,0) -- (8,0);
\draw (8,0) node[right] {$x$};
\draw [->] (0,0) -- (0,4);
\draw (0,4) node[above] {$\Pr$};

\draw (1.3,0) node[below] {$c(\sigma, i)$} node{$\bullet$};
\draw (3.5,0) node[below] {$\tau$} node{$\bullet$};
\draw (7,0) node[below] {$B$} node{$\bullet$};

\draw[thin,->] (1.3,3.6) -- (0.83,{6*exp(-0.8*0.8)});
\draw (1.5,3.6) node[above] {distribution of $\sigma$};

\draw[thin,->] (4.0,3.1) -- (1.58,{1.5*exp(-sqrt((1.5-1.3)^2)/0.4)});
\draw (4,3.) node[above] {Laplace noise distribution ($\beta_{\text{Lap}}$)};

\filldraw [fill=blue!60] 
(3.5,0) -- plot [domain = 3.5:7,samples=200] (\x,{6*exp(-\x*0.8)})
-- (7, 0)   -- cycle;

\draw[dotted,thick] plot [domain = 0.6:7.4,samples=2000] (\x,{6*exp(-\x*0.8)});
\draw plot [domain = -0.5:7.4,samples=200] (\x,{2.5*exp(-sqrt((\x-1.3)^2)/0.4)});

\filldraw [black,fill=blue!60] 
(4.2,1.4) -- (4.2,1.8) -- (4.8, 1.8) -- (4.8,1.4) -- cycle;
\draw (4.8,1.75) node[right] {Outlier error};
\draw (4.8,1.45) node[right] {probability ($\beta_{\text{out}}$)};

\end{tikzpicture}

\caption{Conceptual diagram of our approach. Once the threshold $\tau$ has been determined, there are two sources of errors: error due to outliers and error due to the Laplace noise added to the sum $c(\sigma)$.}
\label{fig:overall-pic}
\end{figure}

\subsection*{Privacy} The overall mechanism is outlined in Algorithm~\ref{algo:global-mech}. The mechanism is $(\epsilon, \delta)$-differentially private, where $\epsilon = \epsilon_1 + \epsilon_2$. Steps 1 and 2 are altogether $\epsilon_1 + \epsilon_2, \delta + 0) = (\epsilon, \delta)$-differentially private due to the basic composition property of differential privacy~\cite{dp-book}. Steps 3 to 6 are $(\epsilon, 0)$-differentially private due to the properties of the BT algorithm. Since the two sub-streams $\sigma(1{:}m)$ and $\sigma(m+1{:}n)$ are disjoint, overall we have $(\epsilon, \delta)$-differential privacy due to the parallel composition property of differential privacy~\cite{par-comp}. Notice that while we release the sum of the first $m$ observations in one step, the mechanism can be modified to release the sum of each of the first $m$ observations through the BT algorithm, if required (with noise scaled as $\approx {\tau \log_2 (m)}/{\epsilon_1}$). 

\begin{algorithm}[!ht]
\SetAlgoLined
\SetAlgorithmName{Mechanism}{mechanism}{List of Mechanisms}
\DontPrintSemicolon{}
\let\oldnl\nl
\newcommand{\nonl}{\renewcommand{\nl}{\let\nl\oldnl}}
\SetKwInOut{Input}{input}
\Input{Input stream $\sigma$, stream length $n$, time lag $m \le n$, privacy parameters $\epsilon$ (split between $\epsilon_1$ and $\epsilon_2$) and $\delta$.}
Estimate $\tau < B$ based on the first $m$ values of $\sigma$ through the mechanism described in Section~\ref{sec:threshold} giving us an $(\epsilon_1, \delta)$-differentially private algorithm.\;
Release $\hat{c}(\sigma, m)$ with the Laplace mechanism with noise scaled to $\frac{\tau}{\epsilon_2}$. \;
\For {$i = m + 1$ to $n$}{
	\If{$\sigma(i) > \tau$}{
		Set $\sigma(i) \leftarrow \tau$.\;
	}
	Use the BT algorithm with noise scaled to $\approx {\tau \log_2 (n - m)}/{\epsilon}$ to release $\hat{c}(\sigma, i)$.\;
}
%
%
\caption{Proposed Global Mechanism}
\label{algo:global-mech}
\end{algorithm}
\section{Privately Estimating the Threshold $\tau$}
\label{sec:threshold}
In this section, we will find how to estimate and then privately release the threshold $\tau$. Ideally, $\tau$ should simultaneously minimize the time lag $m$ and the outlier error $\alpha_{\text{out}}$ (characterized by the probability $\beta_{\text{out}}$). We discarded several straightforward ways of privately computing $\tau$. For instance,

\begin{itemize}
	\item The most obvious choice is the maximum of the $m$ values. To make this differentially private, we need to scale noise according to the sensitivity of the max function. If we use global sensitivity, the estimated threshold $\tau$ will be approximately $B$, resulting in no utility gain. We could instead use smooth sensitivity~\cite{smooth}, but since a possible neighbour of the target stream $\sigma$ may have any value between $0$ and $B$, this would again result in sensitivity close to $B$.
	\item Another alternative is to use the standard deviation of the underlying input distribution $\mathcal{F}_B$ of $\sigma$. However, this requires knowing the distribution in advance. We are interested in a more general problem where only a few simple assumptions about the distribution $\mathcal{F}_B$ hold true and are known beforehand.
%
\end{itemize}

Our statistic of choice is the $p$-quantile (cf. Definition~\ref{def:p-quant}). This can be privately computed using an algorithm similar to the algorithm for computing the median of a sequence using smooth sensitivity~\cite{smooth}. Analogous to Definition~\ref{def:p-quant}, the $p$-quantile of a stream $\sigma \rightarrow_{\mathcal{F}_B} \Sigma$ of $n$ elements is defined as
\[
x_p = \min_{i < n} \left\{ \sigma(i) : | \{ j < n : \sigma(j) < \sigma(i) \} | \ge (1 - p)n \right\}. 
\]
That is, the minimum element of $\sigma$ such that at least a $(1 - p)$ fraction of elements in $\sigma$ are below it. Since the CDF of the input distribution $\mathcal{F}_B$ is unknown in advance, we need to obtain an empirical estimate $\hat{x}_p$ of the $p$-quantile. We shall do so using the first $m$ readings of $\sigma$. For differential privacy, the estimate $\hat{x}_p$ needs to be \emph{stable}.\footnote{A real-valued function $f$ of $\sigma$ is said to be $k$-stable if adding or removing any $k$ elements from $\sigma$ does not change the value of $f$. See~\cite[\S 7]{dp-book}.} From Proposition~\ref{prop:p-theta-m}, this means that we require $m = \Omega({\frac{1}{p}})$ readings. On the other hand, for a continual release application, $m$ should be small compared to $n$. More specifically, the time lag $m$ should satisfy
\begin{equation}
\label{eq:constraint}
n \gg m \gg \frac{1}{p}. 
\end{equation}
Additionally, to minimise outlier error, i.e., to minimise $\beta_{\text{out}}$, $p$ needs to be small.


%
%
%

\subsection{Informal Roadmap}
\label{sub:plan}

%
%
%
Since the empirical $p$-quantile, i.e., $\hat{x}_p$, reaches $x_p$ in expectation, it is not possible to upper bound the probability $\Pr[\hat{x}_{p} < x_p]$ arbitrarily to minimise errors due to truncation (step 5 of Mechanism~\ref{algo:global-mech}). We therefore introduce another parameter $\lambda < 1$, and seek to estimate the $\lambda p$-quantile (instead). This allows us to bound $\Pr[\hat{x}_{\lambda p} < x_p]$ by adjusting the parameter $\lambda$, since $x_{\lambda p} > x_p$ if $\lambda < 1$. We denote this probability bound by $\beta_{\text{qt}}$, shown in Figure~\ref{fig:p-quant}. To make the estimate differentially private, and consequently to use it as the threshold $\tau$, we may use additive Laplace noise.  Due to the properties of the Laplace distribution, $\Pr [ \tau < x_p ]$ is non-zero. We seek to bound this within $\beta_{\text{lt}}$. Finally, to fix $\tau$ below a bound $\tau_{\max}$, we set the upper bound $\beta_{\text{rt}}$ on the probability that $t > \tau_{\max}$. The ability to bound these three error probabilities is important for our utility analysis. In the following, we will formally introduce these sources of errors and will subsequently try to minimise them for utility.

\begin{figure}[h]
\centering

\begin{tikzpicture}
\clip (-0.1,-1) rectangle (9.0,4.1);

\draw[->] (0,0) -- (8,0);
\draw (8,0) node[right] {$x$};
\draw [->] (0,0) -- (0,4);
\draw (0,4) node[above] {$\Pr$};

\draw (1.6,0) node[below] {$x_p$} node{$\bullet$};
\draw (2.3,0) node[below] {$x_{\lambda p}$} node{$\bullet$};
\draw (3.4,0) node[below] {$\hat{x}_{\lambda p}$} node{$\bullet$};
\draw (4.1,0) node[below] {$\tau$} node{$\bullet$};
\draw (4.85,0) node[below] {$\tau_{\max}$} node{$\bullet$};

\draw[thin,->] (1.3,3.6) -- (0.83,{6*exp(-0.8)});
\draw (1.5,3.6) node[above] {distribution of $\sigma$};

\draw[thin,->] (3.0,3.15) -- (2.58,{2.5*exp(-((2.55-2.3)^2)/0.4)});
\draw (3.1,3.05) node[above] {distribution of $\hat{x}_{\lambda p}$};

\draw[thin,->] (5.3,3.55) -- (4.65,{2.5*exp(-((4.6-4.1)^2)/0.5)});
\draw (5.5,3.55) node[above] {distribution of $\tau$};

\filldraw [fill=blue!60] 
(0,0) -- plot [domain = 0:1.6,samples=200] (\x,{2.5*exp(-((\x-2.3)^2)/0.4)})
-- (1.6, 0) -- cycle;

\filldraw [fill=green!30, pattern = vertical lines] 
(0,0) -- plot [domain = 0:3.4,samples=200] (\x,{2.5*exp(-((\x-4.1)^2)/0.5)})
-- (3.4, 0) -- cycle;

\filldraw [fill=green!30, pattern = north east lines] 
plot [domain = 4.85:7.5,samples=200] (\x,{2.5*exp(-((\x-4.1)^2)/0.5)})
-- (4.85, 0) -- cycle;

\draw[dotted,thick] plot [domain = 0.6:7.4,samples=200] (\x,{6*exp(-\x)});
\draw plot [domain = 0:7.4,samples=200] (\x,{2.5*exp(-((\x-2.3)^2)/0.4)});
\draw[dashed] plot [domain = 0:7.4,samples=200] (\x,{2.5*exp(-((\x-4.1)^2)/0.5)});

\filldraw [black,fill=blue!60] 
(5.2,1.4) -- (5.2,1.8) -- (5.8, 1.8) -- (5.8,1.4) -- cycle;
\draw (5.8,1.55) node[right] {$\Pr [ \hat{x}_{\lambda p} < x_p ] \leq \beta_{\text{qt}}$};

\filldraw [black,fill=blue!50, pattern = vertical lines] 
(5.2,2.1) -- (5.2,2.5) -- (5.8, 2.5) -- (5.8,2.1) -- cycle;
\draw (5.8,2.25) node[right] {$\Pr [ \tau < x_p ] \leq \beta_{\text{lt}}$};

\filldraw [black,fill=blue!50, pattern = north east lines] 
(5.2,0.7) -- (5.2,1.1) -- (5.8, 1.1) -- (5.8,0.7) -- cycle;
\draw (5.8,0.85) node[right] {$\Pr [ \tau > \tau_{\max} ] \leq \beta_{\text{rt}} $};

\end{tikzpicture}

\caption{Possible sources of error when estimating the threshold $\tau$ using the $p$-quantile.}
\label{fig:p-quant}
\end{figure}
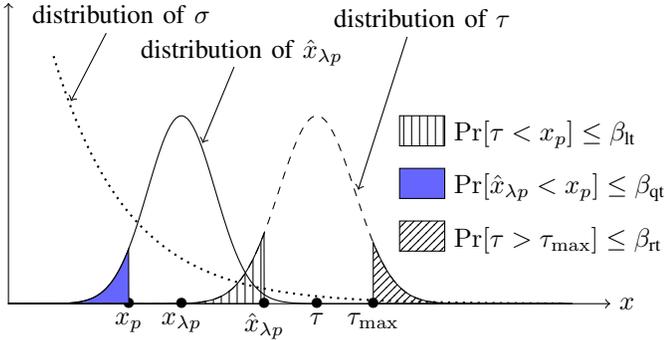

\subsection{Error due to Underestimating the $\lambda p$-Quantile}
\label{sub:comp-quant}


As mentioned above, since the expected value of $\hat{x}_p$ is $x_p$, we cannot bound $\Pr[\hat{x}_p < x_p]$ below and arbitrary bound $\beta_{\text{qt}}$. Thus, instead of estimating the $p$-quantile, we shall estimate the $\lambda p$-quantile with $\frac{1}{pm} < \lambda < 1$. Now the probability of having $\hat{x}_{\lambda p} < x_p$ is given by:
\begin{align}
g(\lambda,p,m) &= \Pr [ \hat{x}_{\lambda p} < x_p ] \nonumber \\
							&= \Pr [ < \lambda p m \text{ values of } \sigma \text{ are } \ge x_{p} ] \nonumber \\
							& = \sum_{i=0}^{\lambda p m} \binom{m}{i} p^i (1-p)^{m-i} \label{eq:func:g}
\end{align}
We denote by $\beta_{\text{qt}}$ the bound on the error probability function $g$.


\subsection{Privately Obtaining the $\lambda p$-Quantile}
For this section we assume that $x_p \le \hat{x}_{\lambda p}$ holds. 
As discussed before, setting the threshold $\tau$ to $\hat{x}_{\lambda p}$ is not private. To obtain a differentially private estimate, we utilize smooth sensitivity.
As shown in~\cite{smooth}, smooth sensitivity can be used to display the median in a differentially private manner. We modify the median algorithm described therein to privately release the $\lambda p$-quantile. First we compute the smooth sensitivity of the empirical $\lambda p$-quantile, i.e., $\hat{x}_{\lambda p}$, as
\begin{align*}
\mathsf{SS}_{\sigma, b}(\hat{x}_{\lambda p}) &=  \max_{k = 0,1,\ldots, m + 1} \{e^{-bk}\mathsf{LS}_{\sigma'}(\hat{x}_{\lambda p}) : d(\sigma,\sigma') \leq k\},
\end{align*}
where 
\[
\mathsf{LS}_{\sigma'} (\hat{x}_{\lambda p}) =  \max_{t = 0,1, \ldots, k+1} | \ddot{\sigma} (P+t ) - \ddot{\sigma} (P+t-k-1) |.
\]
Here, $\ddot{\sigma}$ is the sorted string of the first $m$ values of $\sigma$ in ascending order with $0$ added as a prepend and $B$ as an appendix; and $P$ is the rank of $\hat{x}_{\lambda p}$. This can be done in $O(m^2)$ time~\cite[\S 3.1]{smooth}.

\subsubsection*{Warm:} After computing the smooth sensitivity, we can set the threshold $\tau$ as 
\[ 
\tau = \hat{x}_{\lambda p} +  \frac{\mathsf{SS}_{\sigma, b}(\hat{x}_{\lambda p})}{a} \cdot \text{noise},
\]
where \text{noise} is either the Laplace or standard Gaussian noise. For both, we can set $b \le \frac{\epsilon}{-2\log(\delta)}$ as the smoothing parameter (Definition~\ref{def:smooth}. If we use the Laplace distribution with scale 1, $a = \frac{\epsilon}{2}$ results in $(\epsilon, \delta)$-differential privacy. When the noise is standard Gaussian, then $a = \frac{\epsilon}{\sqrt{-\ln \delta}}$ gives us $(\epsilon, \delta)$-differential privacy~\cite{smooth}. However, as discussed in Section~\ref{sub:plan}, we require $\Pr [ \tau < x_p ]$ to be bounded by an arbitrary $\beta_{\text{lt}}$ (see Figure~ref{fig:p-quant}). Unfortunately, the only way to bound this probability is by adjusting the privacy parameter $\epsilon$, which cannot be arbitrarily chosen (without compromising privacy). Thus, we need to slightly change the above estimate. 

\subsubsection*{Warmer:}  Let $G_\text{ns}$ denote the CDF of the noise distribution. Then instead of the above we can set $\tau$ to 
\begin{align*}
\tau &= \hat{x}_{\lambda p} + \frac{\mathsf{SS}_{\sigma, b} (\hat{x}_{\lambda p})}{a} \cdot ( \text{noise} + \text{offset} ) \\
		&= \hat{x}_{\lambda p} + \frac{\mathsf{SS}_{\sigma, b} (\hat{x}_{\lambda p})}{a} \cdot \text{offset}+ \frac{\mathsf{SS}_{\sigma, b} (\hat{x}_p)}{a} \cdot \text{noise} .
\end{align*}
where $\text{offset} = G_{\text{ns}}^{-1}(1-\beta_{\text{lt}})$.  That is, we offset the noise to the right of $x_p$ to ensure that the probability of the threshold $\tau$ falling below $x_p$ is bounded by $\beta_{\text{lt}}$. Unfortunately, this is no longer differentially private, since the offset itself may leak information. 

\subsubsection*{Solution:}
Informally, to solve this problem we will multiplicatively increase $\mathsf{SS}_{\sigma, b}(\hat{x}_{\lambda p})$ by a factor $\kappa$ such that it ``hides'' the offset as well.\footnote{Thus aiming for a smooth upper bound on smooth sensitivity, rather than smooth sensitivity itself.} First define $\tau'$ to be
\begin{equation}
\label{eq:tau-dash}
\tau' = \hat{x}_{\lambda p} +  \frac{\kappa \mathsf{SS}_{\sigma, b}(\hat{x}_{\lambda p})}{a} \cdot G_{\text{ns}}^{-1}(1-\beta_{\text{lt}}),
\end{equation}
where $\kappa$ is a positive real number to be determined. Then, we finally set
\begin{equation}
\label{eq:final-tau}
\tau = \tau' + \frac{\mathsf{SS}_{\sigma, b} (\tau')}{a} \cdot  \text{noise} 
\end{equation}
We seek the smallest $\kappa$ such that the probability of having $\tau < x_p$ is bounded by $\beta_{\text{lt}}$. Now, to bound this probability a little algebraic manipulation using Eqs.~\ref{eq:tau-dash} and~\ref{eq:final-tau}, together with our assumption $\hat{x}_{\lambda p} \ge x_p$ (beginning of this section), shows that having
\begin{align*}
\Pr \left[ \text{noise} < - \frac{\kappa \mathsf{SS}_{\sigma, b}(\hat{x}_{\lambda p}) G_{\text{ns}}^{-1}(1-\beta_{\text{lt}})}{\mathsf{SS}_{\sigma, b} (\tau')} \right] \leq \beta_{\text{lt}},
\end{align*}
suffices. This is equivalent to
%
%
\begin{align}
G_\text{ns} \left( - \frac{\kappa \mathsf{SS}_{\sigma, b}(\hat{x}_{\lambda p}) G_{\text{ns}}^{-1}(1-\beta_{\text{lt}})}{\mathsf{SS}_{\sigma, b} (\tau')} \right) &\leq \beta_{\text{lt}} \nonumber \\
 \Leftrightarrow  - \frac{\kappa \mathsf{SS}_{\sigma, b}(\hat{x}_{\lambda p}) G_{\text{ns}}^{-1}(1-\beta_{\text{lt}})}{\mathsf{SS}_{\sigma, b} (\tau')}  &\leq G_\text{ns}^{-1} (\beta_{\text{lt}}) \nonumber \\
\Leftrightarrow  \frac{\kappa \mathsf{SS}_{\sigma, b}(\hat{x}_{\lambda p}) G_{\text{ns}}^{-1}(1-\beta_{\text{lt}})}{\mathsf{SS}_{\sigma, b} (\tau')}  & \geq G_\text{ns}^{-1} (1 - \beta_{\text{lt}}) \nonumber \\
\Leftrightarrow  \kappa \mathsf{SS}_{\sigma, b}(\hat{x}_{\lambda p})  & \geq  \mathsf{SS}_{\sigma, b} (\tau') \label{eq:k-times}
\end{align}
Also, from Eq.~\ref{eq:tau-dash}, using the triangle inequality and homogeneity property of smooth sensitivity~\cite{smooth}, we get
\begin{equation}
\label{eq:sum-sense}
\mathsf{SS}_{\sigma, b} (\tau') \leq \mathsf{SS}_{\sigma, b}(\hat{x}_{\lambda p}) + \frac{\kappa G_{\text{ns}}^{-1}(1-\beta_{\text{lt}}) }{a}  \cdot \mathsf{SS}_{\sigma, b} (\mathsf{SS}_{\sigma, b}(\hat{x}_{\lambda p}))
\end{equation}
From the definition of smooth sensitivity we see that $\forall \sigma, \sigma'$ such that $d(\sigma, \sigma') = 1$, we have 
\[
\mathsf{SS}_{\sigma', b}(c) \leq e^b \mathsf{SS}_{\sigma, b}(c). 
\]
This also allows us to write, $\forall \sigma'$ such that $d(\sigma, \sigma') = 1$
\[
\mathsf{LS}_{\sigma'}( \mathsf{SS}_{\sigma, b}(c) ) \le (e^b - 1) \mathsf{SS}_{\sigma', b}(c).
\]
Now similar to the computation of smooth sensitivity of the median in~\cite{smooth}, we have
\begin{align}
& \mathsf{SS}_{\sigma, b} (\mathsf{SS}_{\sigma, b}(\hat{x}_{\lambda p}))  \nonumber \\
&\leq \max_{k \in \mathbb{N}} \{e^{-bk} \mathsf{LS}_{\sigma'}( \mathsf{SS}_{\sigma, b}(\hat{x}_{\lambda p}) ) : d(\sigma, \sigma') \leq k \}  \nonumber\\
&\le \max_{k \in \mathbb{N}} \{e^{-bk}  (e^b - 1)\mathsf{SS}_{\sigma', b} (\hat{x}_{\lambda p})  : d(\sigma, \sigma') \leq k \} \nonumber \\
&\le \max_{k \in \mathbb{N}} \{e^{-bk}  (e^b - 1) e^{bk} \mathsf{SS}_{\sigma, b} (\hat{x}_{\lambda p})  : d(\sigma, \sigma') \leq k \}  \nonumber\\
&\le \max_{k \in \mathbb{N}} \{(e^b - 1) \mathsf{SS}_{\sigma, b} (\hat{x}_{\lambda p})  : d(\sigma, \sigma') \leq k \}  \nonumber \\
&\le (e^b - 1) \mathsf{SS}_{\sigma, b} (\hat{x}_{\lambda p}) \label{eq:double-sense}.
\end{align}
Equating Eqs.~\ref{eq:k-times},~\ref{eq:sum-sense} and~\ref{eq:double-sense} gives us the required $\kappa$ as
\begin{equation}
\label{eq:kappa}
\kappa =  \left(1-\frac{(e^{b}-1)G_{\text{ns}}^{-1}(1-\beta_{\text{lt}})}{a} \right)^{-1}
\end{equation}
Now, putting Eq.~\ref{eq:final-tau} into Eq.~\ref{eq:tau-dash}, and using Eq.~\ref{eq:k-times} we have
\begin{align}
 \tau &= \hat{x}_{\lambda p} +  \frac{\kappa \mathsf{SS}_{\sigma, b}(\hat{x}_{\lambda p})}{a} \cdot G_{\text{ns}}^{-1}(1-\beta_2) + \frac{\mathsf{SS}_{\sigma, b} (\tau')}{a} \cdot  \text{noise} \nonumber \\
 		&\le \hat{x}_{\lambda p} + \frac{\kappa \mathsf{SS}_{\sigma, b}(\hat{x}_{\lambda p})}{a} \cdot G_{\text{ns}}^{-1}(1-\beta_2) + \frac{\kappa \mathsf{SS}_{\sigma, b}(\hat{x}_{\lambda p})}{a} \cdot  \text{noise}  \nonumber \\
 		& = \hat{x}_{\lambda p} + \frac{\kappa \mathsf{SS}_{\sigma, b}(\hat{x}_{\lambda p})}{a} \cdot ( \text{noise} + G_{\text{ns}}^{-1}(1-\beta_{\text{lt}}) ) \label{eq:tau-express},
\end{align}
where $\kappa$ is given by Eq.~\ref{eq:kappa}. The threshold $\tau$ released via the above mechanism is differentially private since $\kappa \mathsf{SS}_{\sigma, b}(\hat{x}_{\lambda p})$ is a smooth upper bound of $\hat{x}_{\lambda p}$ and $\kappa$ only depends on public parameters. 

\subsection{Upper Bound on the Threshold}
\label{max-laplacian}
For our utility analysis in the next section, we require an upper bound on the random variable $\tau$ from Eq.~\ref{eq:tau-express}. We see that with probability at least $1 - \beta_{\text{rt}}$, we have
\begin{equation}\label{eq:tau-max}
\tau \le \tau_{\max} = \hat{x}_{\lambda p} + \frac{\kappa \mathsf{SS}_{\sigma, b}(\hat{x}_{\lambda p})}{a} \cdot (G_{\text{ns}}^{-1}(1-\beta_{\text{lt}}) + G_{\text{ns}}^{-1}(1-\beta_{\text{rt}}))
\end{equation}
%
%

%
%
%
%
%

%
%

\section{Utility Analysis}
\label{sub:error}

For this section, we assume that the $\lambda p$-quantile has been obtained after $m$ steps and satisfies the constraint $x_p \le \hat{x}_{\lambda p}$. Furthermore, a threshold $\tau$ has been obtained via Eq.~\ref{eq:tau-express} satisfying Eq.~\ref{eq:tau-max}. As described in Section~\ref{sub:overview}, our mechanism (Mechanism~\ref{algo:global-mech}) then releases $c$ on every new observation from $m + 1$ to $n$. For the $i$th observation $x = \sigma(i)$ where $i > m$, if $x \le \tau$ then we release $c(\sigma, i)$ through the BT algorithm with noise $\text{Lap}(\frac{\tau \log_2(n-m)}{\epsilon})$. This causes an additive error $\alpha_{\text{Lap}}$ in the computation of $c$ with an associated error probability $\beta_{\text{Lap}}$. On the other hand, if $x > \tau$, we instead assume that the new observation is exactly $\tau$ and then again add noise as before. This induces an additional error term, which we have called outlier error, denoted $\alpha_{\text{out}}$. We denote the probability of the outlier error by $\beta_{\text{out}}$. In the following, we bound these two errors by first assuming the (unrealistic) worst case scenario, i.e., every new observation after the time lag $m$ steps is exactly $B$ with probability $p$. We then use the more realistic assumption that the distribution of the stream is light-tailed, and show that based on real-world datasets we are expected to gain significant utility in practice. 

\subsection{Worst Case Error}
\label{subsub:worst-case}
%

Let $\xi$ denote the PDF of the outlier error and $\Xi$ its CDF. Let $E$ be a random variable denoting the outlier error and let $E_i = \sigma(i) - \min \{ \sigma(i), \tau \}$ denote the outlier error of observation $i$, which is bounded by $B- \tau$. Assuming each element of $\sigma$ is distributed as $X \sim \mathcal{F}_B$ (cf. Definition~\ref{def:dist-stream}), we have $\Xi(x) \geq F(x - x_p)$ for strictly positive $x \in X$. The worst case is when the PDF is given as
\[
\xi(x) = \Delta(x) (1 - p) + \Delta (x - (B - \tau)) p
\]
where $\Delta$ is the Dirac delta function. This means that beyond the $p$-quantile, all the values are equal to $B$. With this assumption we can estimate the $(\alpha_{\text{out}}, \beta_{\text{out}})$-utility (Definition~\ref{def:utility}) as follows:
\begin{align*}
\Pr \left[ \sum_{i=1}^n E_i \geq \alpha_{\text{out}} \right] &= \Pr \left[ h \sum_{i=1}^n E_i \geq h\alpha_{\text{out}}\right]  \\
&=  \Pr \left[ \exp\left( h \sum_{i=1}^n E_i \right) \geq \exp(h\alpha_{\text{out}}) \right] \\
&\leq  \frac{\mathbb{E} [ \exp( h \sum_{i=1}^n E_i ) ]}{\exp(h\alpha_{\text{out}})}  \\
&=  \frac{\prod_{i=1}^n \mathbb{E} [ \exp( h E_i ) ]}{\exp(h\alpha_{\text{out}})}  \\
&=  \frac{\mathbb{E} [ \exp( h E ) ]^n }{\exp(h\alpha_{\text{out}})}  \\
&  = \frac{(1-p+pe^{h(B-\tau)})^n}{e^{h\alpha_{\text{out}}}} = \beta_{\text{out}} 
\end{align*}

%
%
%
%
%
%
Solving for $\alpha_{\text{out}}$, we get
\[
\alpha_{\text{out}} = \frac{n \ln(1-p+pe^{h(B-\tau)}) + \ln\frac{1}{\beta_{\text{out}}} }{h}
\]
The value of $h \approx \frac{1}{(B-\tau)pn}$ minimizes $\alpha_{\text{out}}$. Recall that according to Eq.~\ref{eq:constraint}, we want $m \gg \frac{1}{p}$, which implies that $h \ll 1$. The above then becomes
\[
\alpha_{\text{out}} = pn(B-\tau)\left(\ln\frac{1}{\beta_{\text{out}}}+1\right) + o(1)
\]
Adding this to the utility term $\alpha_{\text{Lap}}$ from the BT algorithm (Eq.~\ref{eq:bt-alpha})~\cite{cont-release} we see that the overall error $\alpha$ is
\begin{align}
\alpha &\le \frac{1}{\epsilon}(\log_2(n-m))^{1.5} \tau \sqrt{8\ln \frac{1}{\beta_{\text{Lap}}}} + (B-\tau)pn \left(\ln \frac{1}{\beta_{\text{out}}} +1\right), \label{eq:alpha-worse}
\end{align}
with probability at most $\beta$, where $\beta$ is a bound on the sum of the five error probabilities.\footnote{i.e.,  ${\beta_\text{qt}}$, ${\beta_\text{lt}}$, ${\beta_\text{Lap}}$, ${\beta_\text{out}}$ and ${\beta_\text{rt}}$.}
If the error is dominated by the first summand, then this leads to an improvement factor of $\frac{B}{\tau}$ in utility over the application of the BT algorithm without our mechanism (see Eq.~\ref{eq:bt-alpha}). However, looking at the second summand, we see that the outlier error is proportional to $pn$. To make this into a constant error term , we need $p \approx \frac{1}{n}$. But recall from Eq.~\ref{eq:constraint} that we require $m \gg \frac{1}{p}$. Thus, it is not possible to bound this error term. Hence, if the input stream has the worst-case distribution, our mechanism does not improve utility. However, arguably, real-world data streams are not distributed in this way.

%


\subsection{Error on Light-Tailed Distributions}
\label{subsub:act-err}
As shown in Section~\ref{sub:datasets}, many real-world data distributions are expected to be light-tailed (cf. Definition~\ref{def:light-tail}), thus behaving significantly differently than the worst case. More precisely, we focus on distributions that are light-tailed beyond their ${p_{\max}}$-quantile, a quantity to be determined shortly. Clearly, this holds true for any $p \le p_{\max}$ as well. Figure~\ref{fig:train-sub-exp} shows that this assumption holds for the train trips dataset for $p = 0.005$-quantile. The figure shows the ECDF of travel times against the CDF of the exponential distribution with parameter $\gamma = \frac{-\ln p}{x_p} = \frac{-\ln 0.005}{x_{0.005}}$ (cf. Fact~\ref{fact:quant-func}). The assumption also holds for the supermarket dataset with the same $p$-quantile. We omit the graph due to repetition.

\begin{figure}[!h]
\centering
\includegraphics[trim={0 0 0 4mm},clip]{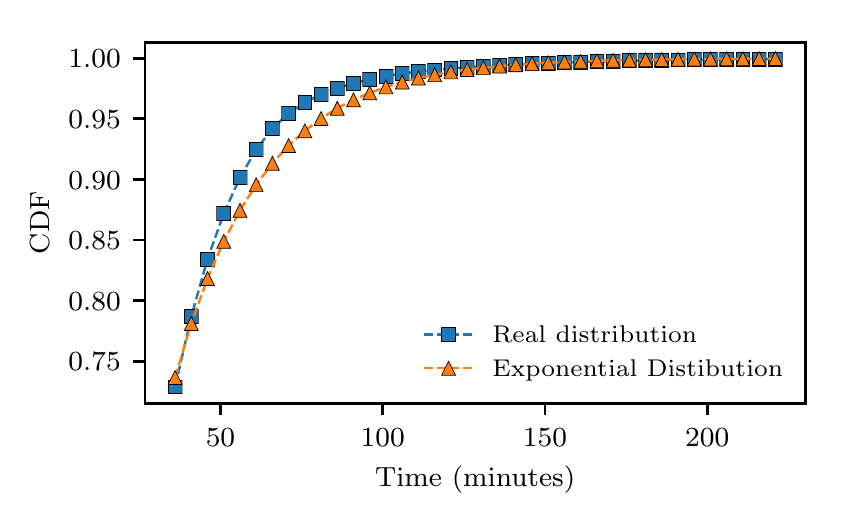}
\caption{The distribution of the train trips dataset compared to an exponential distribution with the same $p = 0.005$-quantile. The real distribution is above the exponential distribution, indicating that it is light-tailed.}
\label{fig:train-sub-exp}
\end{figure}

%
%

Since the distribution is light-tailed, we can use Proposition~\ref{prop:pr} and the assumption $x_p \le \hat{x}_{\lambda p}$ to conclude that for all $p \le p_{\max}$,
\begin{equation}
\label{eq:pr-thru-p}
\hat{x}_{\lambda p} \cdot r \geq x_{p^r}, \forall r \geq 1.
\end{equation}
Thus, instead of using the threshold $\tau$ directly from Eq.~\ref{eq:final-tau}, we multiply it by $r$ and set it as the threshold. According to the above equation, this results in reduced outlier error whenever $r > 1$. 


\subsubsection*{Determining $p_{\max}$:}
To determine the value of $p_{\max}$, we perform a series of experiments on the train trips and the supermarket dataset. We seek a value of $p_{\max}$ that ensures the light-tailed property on both datasets. We vary $p^r$ beginning from a value of $0.1$ to increasingly small values. Against each value of $p^r$ we obtain the empirical $p^r$-quantile, i.e., $\hat{x}_{p^r}$. We then use different values of $p$, e.g., 0.1, 0.01, and so on. From each pair of values of $p$ and $p^r$, we obtain $r$ and multiply it with the empirical $p^r$-quantile to obtain $r \hat{x}_{p}$. The aim is to find a value of $p_{\max}$ such that for all $p \le p_{\max}$, $r \hat{x}_{p} \approx \hat{x}_{p^r}$. Since $\hat{x}_{\lambda p} \ge \hat{x}_{p}$, this implies that  Eq.~\ref{eq:pr-thru-p} would be satisfied. The results are shown in Figure~\ref{fig:quant_trains} for the train trips dataset and Figure~\ref{fig:quant_sm} for the supermarket dataset. The results suggest that $p_{\max} \approx 0.005$ suffices. 

%


\begin{figure*}[!h]
    \centering
  \subfloat[Train trips dataset]{%
 	 \label{fig:quant_trains}
       \includegraphics[width=0.48\textwidth, trim={0 0 0 4mm},clip]{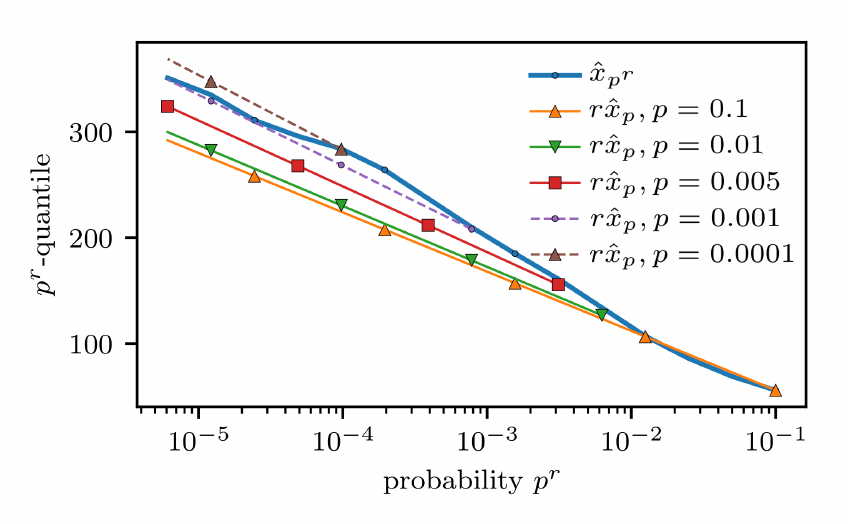}}
  \subfloat[Supermarket dataset]{%
   	\label{fig:quant_sm}
        \includegraphics[width=0.48\textwidth, trim={0 0 0 4mm},clip]{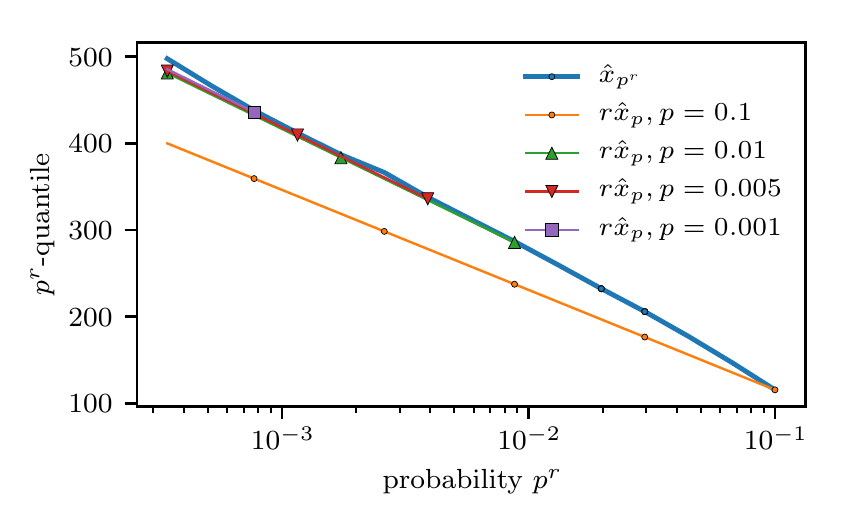}}
  \caption{Error in estimating the empirical $p^r$-quantile through empirical $p$-quantile with different choices of $p$. We see that below $p_{\max} \approx 0.005$, the input datasets satisfy $\hat{x}_{p} \cdot r \approx x_{p^r}$. }
  \label{fig:quants} 
\end{figure*}


\subsubsection*{Error Bound:} Now, assuming that the input stream is light-tailed we see that the outlier error is bounded by the properties of the exponential distribution. That is, PDF $\xi$ of the outlier error can be written as 
\[
\xi (x) = \Delta(x)(1-p^r) + p^r \cdot \gamma e^{-\gamma x},
\]
where $\gamma = \frac{-\ln p}{x_p}$. Thus, 
\begin{equation*}
\Pr \left[ \sum_{i=1}^n E_i \geq \alpha_{\text{out}} \right] \leq  \frac{\mathbb{E} [ \exp( h E ) ]^n }{\exp(h\alpha_{\text{out}})} 
	\leq \frac{(1 + \frac{h}{\gamma - h})^n}{e^{h\alpha_{\text{out}}}} = \beta_{\text{out}}.
\end{equation*}
This gives us,
\[
\alpha_{\text{out}} = \frac{n \ln(1 + \frac{h}{\gamma - h}) + \ln(\frac{1}{\beta_{\text{out}}})}{h}
\]
We can choose the $h$ that minimises $\alpha_{\text{out}}$ as 
\[
h = {\gamma} \left( {\sqrt{\frac{p^r \cdot n}{\ln\frac{1}{\beta_{\text{out}}}}}+1} \right)^{-1},
\]
which leads to 
\[
\alpha_{\text{out}} \leq \frac{-x_p}{\ln p}\left(\sqrt{p^r \cdot n}+\sqrt{\ln\frac{1}{\beta_{\text{out}}}}\right)^2.
\]
Adding this to the error term $\alpha_{\text{Lap}}$ from the BT algorithm (Eq.~\ref{eq:bt-alpha}) and using the assumption $x_p \le \hat{x}_{\lambda p}$, we see that the overall error $\alpha$ is
\begin{align}
\alpha &\le \frac{1}{\epsilon}(\log_2(n-m))^{1.5} \tau r \sqrt{8\ln \frac{1}{\beta_{\text{Lap}}}} \nonumber \\
&+ \frac{-\hat{x}_{\lambda p}}{\ln p}\left(\sqrt{p^r \cdot n}+\sqrt{\ln\frac{1}{\beta_{\text{out}}}}\right)^2  \label{eq:utility} 
\end{align}
with probability at least $1 - \beta$, where $\beta$ is once again a bound on the five error probabilities. Now, to bound the second error term (the second summand) by a constant, we require $p^r \approx {\frac{1}{n}}$. Thus an $r$ logarithmic in $n$ suffices.
With this value of $r$ we see that the overall error is bounded by $O(\tau (\log_2 n)^{1.5}/\epsilon)$. Thus, we obtain an improvement factor of $B/\tau$ over the BT algorithm, which was the aim of our mechanism. In the next section, we will show how to optimize the parameters for utility. 


\section{Optimizing Utility}
\label{sec:opt-util}
\subsection{Optimized Parameters}
To optimize $\alpha$ given by Eq.~\ref{eq:utility}, we ran a series of experiments on the two datasets using the Python library SciPy.\footnote{\url{https://www.scipy.org/}} Specifically, we used a truncated Newton method~\cite{tnc} (TNC) implemented by the \texttt{scipy.optimize.minimize} method to optimize $\alpha$. We fix $n = \text{25,000,000}$ for the train trips dataset and $n = \text{150,000}$ for the supermarket dataset. The parameter $\beta$, i.e., overall probability of exceeding an error of $\alpha$, was fixed to 0.02, and $\delta$ was fixed to {$2^{-20}$}.\footnote{While this value of $\delta$ is higher than recommended (i.e., negligible in $1/n$~\cite[\S 2.3, p. 18]{dp-book}), lower values, say $2^{-30}$~\cite[\S 3, p. 5]{psi}, have a minor impact on utility in our experiments.} For both datasets, we analyze the influence of local and global parameters separately. 




\subsubsection*{Effect of Local Parameters}
We fixed $\epsilon = 1$ for this series of experiments. Then, for different values of the time lag $m$, we ran the optimizer on the objective function $\alpha$ given by Eq.~\ref{eq:utility}, with the constraints: $p \le p_{\max} = 0.005$, $\lambda \leq 1$, $r \geq 1$, and $\kappa > 0$. 
Note that the optimization algorithm is deterministic: given fixed global parameters, we obtain the same value of the local parameters each time. We also define the improvement factor (IF), as the ratio of error obtained from the BT algorithm to the error obtained through our mechanism. The results are shown in Tables~\ref{tab:train} and \ref{tab:sm}. 

\begin{table}[!ht]
\centering
\caption{Optimized parameters for the train trips dataset}
\label{tab:train}
\tabcolsep=0.11cm
\begin{tabular}{l|l|l|l|l|l|l|l|l|l|l}

$m$    & IF & $r$  & $\lambda$ & $p$    & $\frac{\epsilon_1}{\epsilon}$ & $\frac{\beta_\text{qt}}{\beta}$ & $\frac{\beta_\text{lt}}{\beta}$ & $\frac{\beta_\text{Lap}}{\beta}$ & $\frac{\beta_\text{out}}{\beta}$ & $\frac{\beta_\text{rt}}{\beta}$ \\
\hline\hline

40000  & 1                                                            & 1    & 1        & 0.005  & 0                             & 0.00                & 0.00                      & 1                       & 0.00                       & 0.00                      \\

50000  & 2.32                                                            & 1.63    &  0.85       & 0.005  & 0.8                             & 0.09                      & 0.3                      & 0.41                       & 0.09                       & 0.09                      \\

60000  & 2.8                                                          & 1.49 & 0.83     & 0.0049 & 0.9                           & 0.13                      & 0.23                      & 0.37                       & 0.13                       & 0.13                      \\

100000 & 3.69                                                         & 1.76 & 0.87     & 0.0048 & 0.9                           & 0.07                      & 0.1                       & 0.68                       & 0.07                       & 0.07                      \\

300000 & 4.34                                                         & 1.91 & 0.87     & 0.005  & 0.79                          & 0.11                      & 0.11                      & 0.57                       & 0.11                       & 0.11  \\
                   
\end{tabular}

\caption{Optimized parameters for the supermarket dataset}
\label{tab:sm}
\tabcolsep=0.11cm
\begin{tabular}{l|l|l|l|l|l|l|l|l|l|l}

$m$    & IF & $r$  & $\lambda$ & $p$    & $\frac{\epsilon_1}{\epsilon}$ & $\frac{\beta_\text{qt}}{\beta}$ & $\frac{\beta_\text{lt}}{\beta}$ & $\frac{\beta_\text{Lap}}{\beta}$ & $\frac{\beta_\text{out}}{\beta}$ & $\frac{\beta_\text{rt}}{\beta}$ \\
\hline\hline

40000  & 1                                                            & 1    & 1        & 0.005  & 0                             & 0.00                & 0.00                      & 1                       & 0.00                       & 0.00                      \\

50000  & 4.23                                                            & 1    &  0.81       & 0.005  & 0.82                             & 0.23                      & 0.19                      & 0.19                       & 0.19                       & 0.19 \\

60000  & 4.86                                                          & 1.08 & 0.82     & 0.0049 & 0.88                           & 0.15                      & 0.27                      & 0.27                       & 0.15                       & 0.15                      \\

100000 & 6.67                                                         & 1.11 & 0.85     & 0.0046 & 0.81                           & 0.07                      & 0.07                       & 0.71                       & 0.07                       & 0.07                      \\

                   
\end{tabular}
\end{table}
We see that as the time lag $m$ grows, IF is less impacted by Laplace noise added to the sum as indicated by the decreasing ratio $\beta_\text{Lap}/\beta$.
The improvement factor grows with increasing $m$, but this is essentially a trade-off between releasing or withholding the sum. The rest of the parameters are relatively stable, with higher values of $r$ indicating that the threshold can be set higher than the estimated experimental $\lambda p$-quantile. For smaller $m$, we do not see any improvement in utility over the basic BT algorithm. Note that in such a case our mechanism simply releases the sum using the BT algorithm with noise scaled to $B$. Thus, we do not incur any extra cost in utility.

\subsubsection*{Effect of Global Parameters}
The parameters $\epsilon$, $\delta$, $\beta$, and $m$ are global parameters specified to the optimization algorithm. The parameters $\beta$ and $\delta$ are fixed as before.  Thus, we look at the evolution of IF with different values of $\epsilon$ and $m$. For each value, we run the optimizer to output a set of local parameters that maximize utility. For this, we only use the train trips dataset with $n = \text{25,000,000}$.


\begin{itemize}
\item [$\epsilon$:] Smooth sensitivity is roughly proportional to $\frac{1}{\epsilon_1^2}$, so smaller values of $\epsilon$ (and consequently $\epsilon_1$) will not result in a high IF. On the other hand, if $\epsilon$ is too large, the error caused by truncation (step 5 of Mechanism~\ref{algo:global-mech}) will overwhelm the noise due to the Laplace mechanism, and hence the IF will be low. Figure~\ref{fig:eps-if} (left) shows this trend, where we plot the improvement factor of our algorithm over the BT algorithm by fixing $m = \text{50,000}$. 

\item [$m$:] The impact of the time lag $m$ is data dependent. We do not see much improvement when $m$ is small, say around 10,000. With $m$ around 50,000 we see noticeable increase in IF. This is indicated by Figure~\ref{fig:eps-if} (right), where we have fixed $\epsilon = 1$.
\end{itemize}


\begin{figure}[!ht]
\centering
\includegraphics[trim={0 0 0 1.5mm},clip]{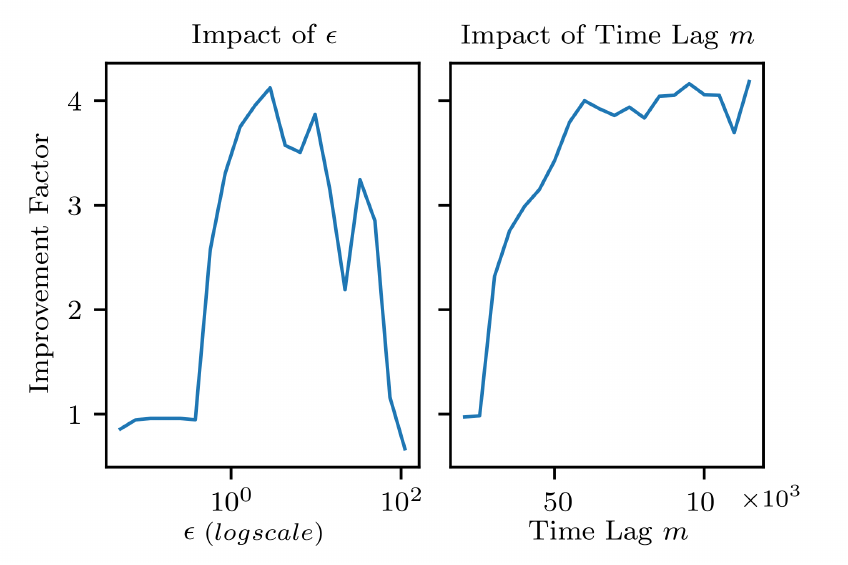}
\caption{Influence of global parameters on the improvement factor.}
\label{fig:eps-if}
\end{figure}

\subsection{Private Choice of Parameters}
The parameters required as input in Mechanism~\ref{algo:global-mech} are $\epsilon_1$, $\epsilon_2$, $\delta$, $m$ and $r$. For reasons of privacy, the choice of these parameters cannot be based on optimization on a particular input dataset. We therefore discuss some heuristic choices for these parameters based on our experiments above. The parameters $\epsilon = \epsilon_1 + \epsilon_2$ and $\delta$ can be chosen in the standard way. For instance, $\epsilon = 1$ and $\delta = n^{-2}$. From Tables~\ref{tab:train} and \ref{tab:sm}, a value in the range $0.8$ to $0.9$ is a reasonable choice for $\epsilon_1$. Note that $\epsilon_2$ is readily determined by $\epsilon$ and $\epsilon_1$. From the same tables, we see that an $r$ between 1 and 2 suffices. We therefore discuss heuristics for choosing $m$.

\subsubsection*{Heuristics for Selecting the Time Lag $m$:}
We specify two criteria that should be satisfied by the time lag. By assigning reasonably conservative (with respect to utility guarantees) values to the free parameters in the two criteria, we obtain a value of $m$ that is expected to provide good utility in practice. The overall value of $m$ can be obtained as the maximum of the two values returned by the criteria. 

\paragraph{First Criterion} This is the probability of having $\hat{x}_{\lambda p} < x_p$ given by the function $g$ in Eq.~\ref{eq:func:g}. We set $g(\lambda, p, m) = g(0.5, p_{\max}, m) < \beta$. All our optimization experiments returned a value of $\lambda$ close to 1. Thus, setting $\lambda = 0.5$ is a reasonably conservative choice. 

\paragraph{Second Criterion} This is related to the (approximate) scale of smooth sensitivity:\footnote{If the time lag is small, the observations will be sparsely distributed implying a large scale of smooth sensitivity.}
\begin{equation}
\frac{\kappa \mathsf{SS}_{\sigma, b}(\hat{x}_{\lambda p})}{a} \cdot G_{\text{ns}}^{-1}(1-\beta) \approx \frac{B}{10} \label{eq:approx-ss}
\end{equation}
The term $B/10$ is arbitrarily chosen to ensure that the scale is a few orders of magnitude less than $B$. With the Laplace noise distribution, we can take $ a = \frac{\epsilon}{\sqrt{-\ln{\delta}}} $ and $b = \min(1,\frac{\epsilon}{-2\ln\delta})$. 
This readily gives us $\kappa$ through Eq.~\ref{eq:kappa}. Now, to get a conservative bound on $\mathsf{SS}_{\sigma, b}(\hat{x}_{\lambda p})$, 
we first assume that the exponential distribution has its $p_{\max}$-quantile close to $B$. In other words, $x_{p_{\max}} \approx B$. This means 
that the upper bound on the smooth sensitivity is conservative. Now at $B$, the probability density function of the exponential 
distribution is $\frac{-p_{\max}\ln p_{\max}}{B}$. With $m$ observations, the inverse of the average distance between two observations is therefore roughly 
\begin{equation}
d = \frac{-m p_{\max}\ln p_{\max}}{B}. \label{eq:density}
\end{equation}
Now, assuming that the density is approximately constant in the neighbourhood of $B$, smooth sensitivity is given by
\[ 
\max\limits_{k \in \mathbb{N}}\{d k \exp (-bk) \}
\]
We can bound the above if we replace natural numbers with real numbers, resulting in 
\begin{equation}
\mathsf{SS}_{\sigma, b}(\hat{x}_{\lambda p}) < \frac{d \exp(-1)}{b}. \label{eq:real-ss}
\end{equation}
Now, combining Eqs.~\ref{eq:approx-ss},~\ref{eq:density} and \ref{eq:real-ss} and solving for $m$, we get our second criterion on $m$ as
\[
m > \frac{20 \kappa (-\ln \delta)^{1.5} \exp(-1) G_{\text{ns}}^{-1}(1-\beta)}{-\epsilon^2 p_{\max} \ln p_{\max}} 
\]

With $\beta = 0.02$, $\delta = 2^{-20}$, $\epsilon = 1$ and $p_{\max} = 0.005$, we find that $m$ should be at least 20,000 in order to satisfy the first criterion. The second criterion imposes $m>80,000$ for good utility. From Tables~\ref{tab:train} and \ref{tab:sm} we see that with $m>80,000$ we indeed obtain good utility. Of course, the higher the time lag $m$, the better the utility gain, with the trade-off that there is a longer time lag before we output the sum.

\section{Experimental Evaluation}

\subsection{Accumulated Error on the Sum}
We now show the improvement factor in computing the moving average (sum) through our mechanism. Since the error is maximized at step $n$, i.e., the last observation of the stream, we compare the value of $\hat{c}(\sigma, n)$ through our mechanism against its counterpart via the BT algorithm. For both datasets, we run the two mechanisms a total of 20,000 times and display the empirical probability density function (PDF) of error. 

\subsubsection{Train Trips Dataset}
We fixed $\epsilon = 1$, $\delta = 2^{-20}$, $\beta = 0.02$, $m = 50,000$ and $n = 25,000,000$, and obtained the values of the local parameters after optimization, shown in Table~\ref{tab:train}. We set $B = 1440$ minutes, which is the maximum possible commute time in a 24 hour period. Figure~\ref{fig:pdf_error_train_trips} shows the PDF of the resulting error (normalized by the maximum value $B$) of our mechanism and the BT algorithm. We see that the error in our case is more tightly concentrated around $0$. On average, we obtain an improvement factor of 3.5. 

\begin{figure}[ht]
\centering
\includegraphics[trim={10mm 0 0 0.5mm},clip]{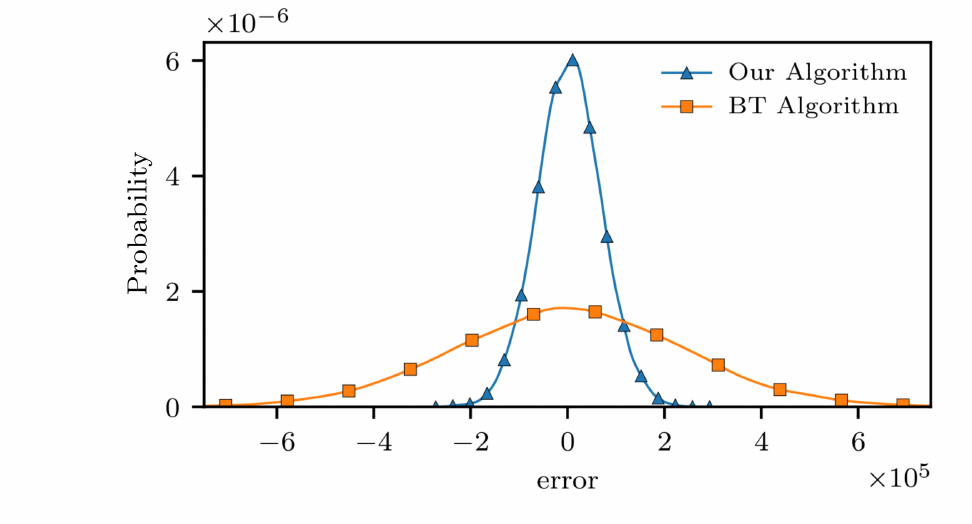}
\caption{PDF of the error on the train trips dataset through our mechanism and the BT algorithm.}
\label{fig:pdf_error_train_trips}
\end{figure}

We are also interested in knowing whether our estimation of the outlier error (i.e., the second summand in Eq.~\ref{eq:utility}), is close to the actual outlier error. This will validate whether our assumption that the distribution of the dataset is light-tailed. To verify this we re-ran our mechanism 20,000 times on the same dataset and obtained the ratio $\frac{\text{real error}}{\text{estimated error}}$ with the same parameters as above. 
Figure~\ref{fig:pdf_threshold_train_trips} shows that we have erred on the precautionary side with our estimation of outlier error being well within the actual error. 


\begin{figure}[ht]
\centering
\includegraphics[trim={0 0 0 4mm},clip]{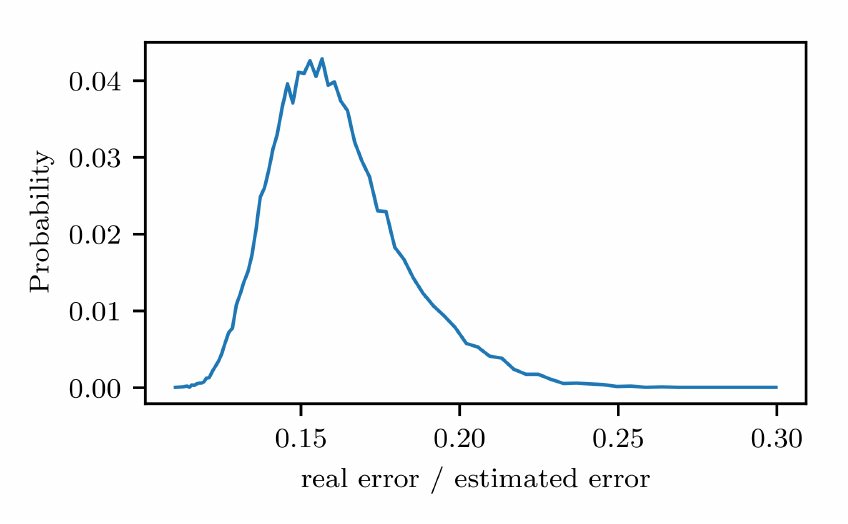}
\caption{Distribution of the outlier error ratio $\text{real error}/\text{estimated error}$ on the train trips dataset. Our estimated outlier error is well below the real outlier error.}
\label{fig:pdf_threshold_train_trips}
\end{figure}

\subsubsection{Supermarket Dataset}
For the supermarket dataset, we use the same set of parameters except that we have $n = 150,000$ (due to less data points) and $B = 3,000$ dollars (a conservative guess on the amount spent). 
Figure~\ref{fig:pdf_error_sm} shows the PDF of the error from our mechanism and the BT algorithm. Once again the error through our mechanism is more tightly concentrated around $0$. For this dataset, we perform much better than the BT algorithm, with an improvement factor of 9 on average.


\begin{figure}[ht]
\centering
\includegraphics[trim={10mm 0 5mm 0.5mm},clip]{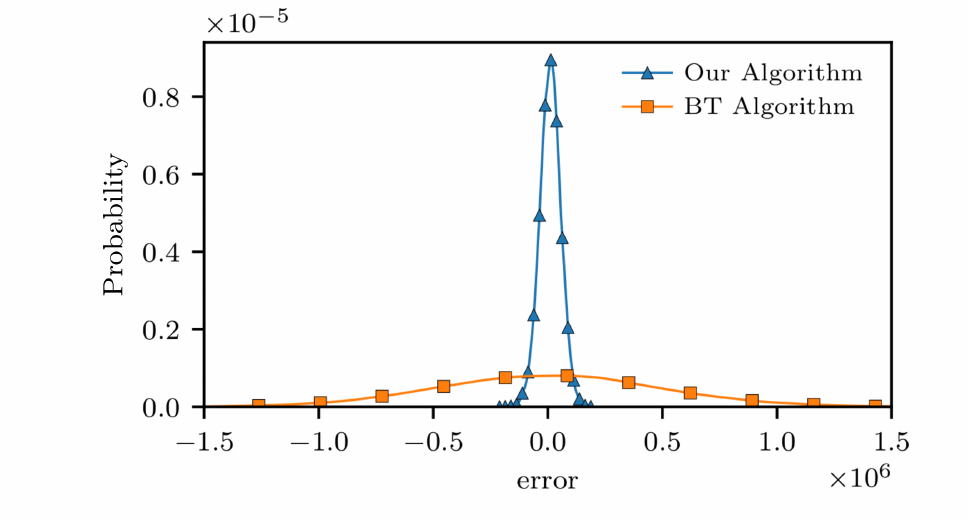}
\caption{PDF of the error on the supermarket dataset through our mechanism and the BT algorithm.}
\label{fig:pdf_error_sm}
\end{figure}

For this dataset as well we are interested in knowing whether the estimated error due to outliers is well below the actual error. Figure~\ref{fig:pdf_threshold_sm} shows that the dataset does indeed have a light-tailed distribution with the actual error being almost always below the estimated value.


\begin{figure}[ht]
\centering
\includegraphics[trim={0 0 0 4mm},clip]{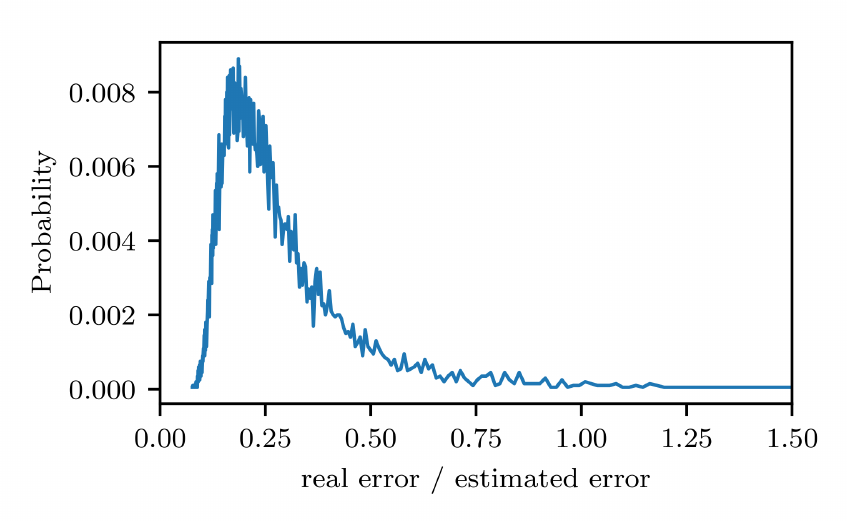}
\caption{Distribution of the outlier error ratio $\text{real error}/\text{estimated error}$ on the supermarket dataset. Once again our estimated outlier error is close to the real outlier error.}
\label{fig:pdf_threshold_sm}
\end{figure}

\subsection{Does the Distribution Remain Light-Tailed across Time?}
Recall that our mechanism promises improved utility based on the premise that data distribution is light-tailed. 
Since the input stream is time dependent, the estimated threshold (using the $\lambda p$-quantile) through $m$ observations with a given time period may be drastically different from its estimate via a different time period. 
%
To ensure that this is not the case, we analyzed the distribution of the train trips dataset across different hours and different days of the week. 
The distributions are shown in Figures~\ref{fig:dist-hours} and~\ref{fig:dist-days}, respectively. While the beginning of the distributions show variation based on the time period, the tails are similar and light-tailed. Thus, our estimated threshold is likely to improve utility independent of the time period in real datasets. 


\begin{figure}[ht]
\centering
\includegraphics[trim={0 0 0 4mm},clip]{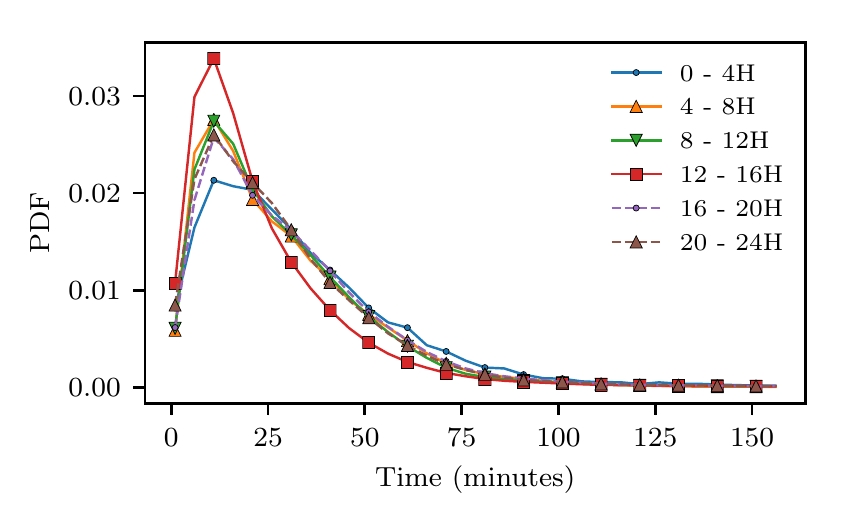}
\caption{Train commute time distribution from different hours of the day. Again all are light-tailed.}
\label{fig:dist-hours}
\end{figure}

\begin{figure}[ht]
\centering
\includegraphics[trim={0 0 0 4mm},clip]{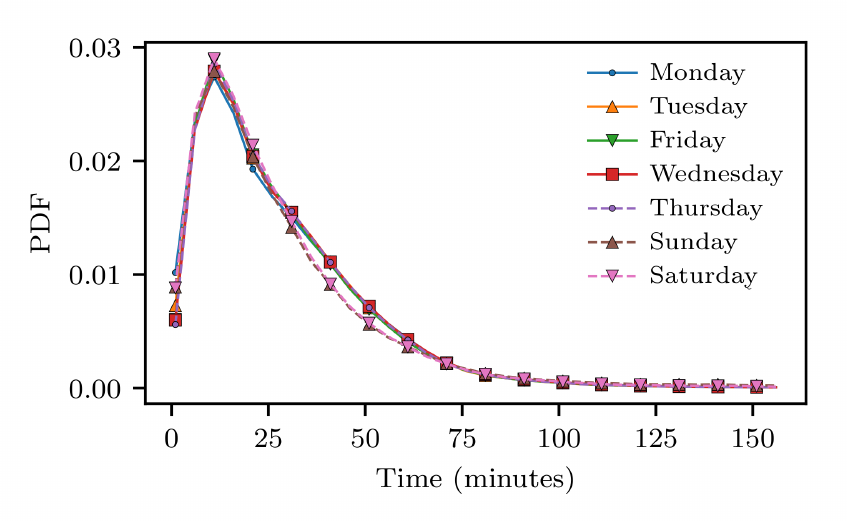}
\caption{Train commute time distribution from different days of the week. All are light-tailed.}
\label{fig:dist-days}
\end{figure}

\section{Related work}
As previously noted, the privacy-preserving algorithms for continual release of statistics from binary streams proposed in~\cite{cont-observe, cont-release} can be generalized to the scenario addressed in this paper, i.e., release of statistics from a stream whose values are from the real interval $[0, B]$. Indeed, we have used the algorithm from~\cite{cont-release} as one of the components of our method. However, the focus of the two works in~\cite{cont-observe} and \cite{cont-release} is on improving the error for binary strings which do not have the added factor of $B$. The two  algorithms are based on event-level privacy. As such, if the aim of privacy is to protect all events from an individual (e.g., all trips made by an individual over the course of the whole year), then the privacy provided by these algorithms is insufficient. The work from~\cite{w-event} attempts to improve this by offering privacy for up to $w$ successive events. Noting that any $w$ successive events might not contain multiple events originating from a single individual, the authors from~\cite{l-trajectory} introduce $l$-trajectory privacy, where any successive $l$ events from a user are targeted for privacy. These works essentially propose privacy mechanisms for variants of the definition of differential privacy where neighbouring streams are defined differently from the standard definition of Hamming distance. We note that our method can be easily used in conjunction with these algorithms, as we only use the BT algorithm from~\cite{cont-release} in a modular way. However, to find a utility maximizing threshold in a differentially private manner for any variation in the definition of neighbouring streams requires tweaking our mechanism. Likewise, these algorithms also target infinite streams as opposed to bounded streams (as is done in our paper). Application of our approach to these settings is an interesting area for future work. 

As argued before, privacy-preserving continual release of the sum is only one example of functions that can be released with improved utility through our mechanism. As long as the target function remains a function of the stream, has reduced sensitivity based on tighter concentration of input data, and the error due to outliers can be bounded and related to the $p$-quantile, we can adapt all the steps of our method to the given function. This allows us to estimate the threshold from the data, thus finding the optimum balance between the error due to symmetric noise and the error due to outliers. Examples of such functions include the sliding window average or the decaying sum where either past observations are completely discarded or are given progressively less weights~\cite{decay-ave}. Another example is continually releasing histogram of the input stream where we would like to completely discard bins above the main concentration of the data.


Our work on estimating the threshold using the $p$-quantile can be thought of as an exercise in finding ``robust'' statistics~\cite{rob-stats-book} with differential privacy, a line of work that was discussed in~\cite{robust-stats} and \cite{dwork2010differential}. These works estimate the scale of the input data with the help of the interquartile range using the Propose-Test-Release approach~\cite{robust-stats}. Briefly, this approach checks whether a given analysis uses a function that is robust or stable~\cite{salil-tut} on a given dataset or not. If the answer is no, the analysis is abandoned. In other words, the interquartile range may not even be released if it is not stable for the given dataset. Our approach is different as we use the $p$-quantile as the estimate of the scale of the input dataset, and use smooth sensitivity to release it. Unlike the Propose-Test-Release approach which may not release the $p$-quantile depending on input data, we have the advantage that we always obtain an estimate. This allows us to optimize utility by bounding errors introduced by the estimation of the scale of the dataset. We note that the problem of finding differentially private quantiles is also tackled in~\cite{smith2011privacy}, but the main ingredient there is the exponential mechanism~\cite{exp-mech}, and the context is static datasets rather than continual release of data. 

The main idea of our work is to reduce the sensitivity of the query (in our case, the moving sum), by relying on some initial knowledge of the input data. If we succeed in reducing sensitivity, we significantly reduce the scale of noise added to the query, thus improving accuracy of the query answer. A similar approach has also been used in some other works for other query types or applications. In~\cite{fan2013differentially}, the authors use prior knowledge of the dataset to release time-series data with better accuracy. The aforementioned approach is also used in~\cite{xu-hist} where the aim is to display a differentially private histogram by altering the size of the bins in order to artificially reduce sensitivity.


\section{Conclusion}
We have presented a privacy-preserving mechanism to continually display the moving average of a stream of observations where the bound on each observation is either too conservative or not known a priori. We have relied on justified assumptions on real-world datasets to obtain a better bound on observations of the stream. Moreover, we have shown how to obtain this bound in a differentially private manner while optimizing utility. Our mechanism can be applied to many real-world applications where continuous monitoring and reporting of statistics is required, e.g., smart meter data and commute times. Our techniques can be improved in several ways. We have relied on the quantile to estimate the bound on the streaming data based on smooth sensitivity. There may be other ways to display the quantile using other robust statistics. Our mechanism can be adapted to compute functions other than the moving average. Likewise, our method can be used in conjunction with algorithms that provide privacy for multiple events instead of single events as is done in this paper. Overall, we see our work as an instance of applying differential privacy in practice.

%
%
%

\bibliographystyle{IEEEtran}
\bibliography{references}

\end{document}